\documentclass[american,aps,prx,superscriptaddress,longbibliography,floatfix]{revtex4-2}
\usepackage[T1]{fontenc}
\usepackage{inputenc}
\usepackage{color}
\usepackage{babel}
\usepackage{physics}
\usepackage{mathtools}
\usepackage{bbding}
\usepackage{pifont}
\usepackage{textcomp}
\usepackage{wasysym}
\usepackage{amsthm}
\usepackage{nicefrac}
\usepackage{wasysym}
\usepackage{dsfont}
\usepackage{amsmath}
\usepackage{bbm}
\usepackage{amssymb}
\usepackage{bbold}
\usepackage{graphicx}
\usepackage{enumitem}
\definecolor{darkblue}{rgb}{0.0, 0.0, 0.55}

\usepackage[colorlinks=true,
            allcolors=blue]{hyperref}

\hypersetup{
     colorlinks   = true,
     linkcolor    = blue,
     citecolor    = blue,
     urlcolor     = blue,
     }

\usepackage{appendix}
\renewcommand{\selectlanguage}[1]{}
\makeatletter
\@ifundefined{textcolor}{}
{%
	\definecolor{BLACK}{gray}{0}
	\definecolor{WHITE}{gray}{1}
	\definecolor{RED}{rgb}{1,0,0}
	\definecolor{GREEN}{rgb}{0,1,0}
	\definecolor{BLUE}{rgb}{0,0,1}
	\definecolor{CYAN}{cmyk}{1,0,0,0}
	\definecolor{MAGENTA}{cmyk}{0,1,0,0}
	\definecolor{YELLOW}{cmyk}{0,0,1,0}
}

\usepackage{babel}
\usepackage{txfonts}
\usepackage{colortbl}\definecolor{myurlcolor}{rgb}{0,0,0.7}
\newcommand{\id}{{\operatorname{id}}}

\usepackage[normalem]{ulem}
\usepackage{tikz}
\usetikzlibrary{circuits.ee.IEC}

\newtheorem{theorem}{Theorem}

\newtheorem{lemma}{Lemma}
\newtheorem*{nonnumberedth}{Theorem}
\newtheorem*{nonnumberedprop}{Proposition}
\newtheorem*{nonumberedlemma}{Lemma}

\newtheorem{proposition}{Proposition}

\def\tr{\operatorname{tr}}
\def\supp{\operatorname{supp}}
\def\St{\operatorname{St}}
\def\Ch{\operatorname{Ch}}
\def\Pos{\operatorname{Pos}}
\def\id{\operatorname{id}}
\def\d{\operatorname{d}}

\newcommand{\n}{\mathcal{N}}
\newcommand{\m}{\mathcal{M}}
\newcommand{\U}{\mathcal{U}}

\usepackage{tikz}

\makeatletter
\newsavebox\myboxA
\newsavebox\myboxB
\newlength\mylenA
\newcommand*\xoverline[2][0.75]{%
    \sbox{\myboxA}{$\m@th#2$}%
    \setbox\myboxB\null
    \ht\myboxB=\ht\myboxA%
    \dp\myboxB=\dp\myboxA%
    \wd\myboxB=#1\wd\myboxA
    \sbox\myboxB{$\m@th\overline{\copy\myboxB}$}
    \setlength\mylenA{\the\wd\myboxA}
    \addtolength\mylenA{-\the\wd\myboxB}%
    \ifdim\wd\myboxB<\wd\myboxA%
       \rlap{\hskip 0.5\mylenA\usebox\myboxB}{\usebox\myboxA}%
    \else
        \hskip -0.5\mylenA\rlap{\usebox\myboxA}{\hskip 0.5\mylenA\usebox\myboxB}%
    \fi}
\makeatother

\usepackage{hyperref}

\begin{document}
\title{Erasure cost of a quantum process: A thermodynamic meaning of the dynamical min-entropy}

\author{Himanshu Badhani}\email{himanshubadhani@gmail.com}
\affiliation{Center for Security, Theory and Algorithmic Research, International Institute of Information Technology, Hyderabad, Gachibowli, Telangana 500032, India}
\affiliation{Centre for Quantum Science and Technology, International Institute of Information Technology, Hyderabad, Gachibowli, Telangana 500032, India}

\author{Dhanuja GS}
\affiliation{Centre for Quantum Science and Technology, International Institute of Information Technology, Hyderabad, Gachibowli, Telangana 500032, India}
\affiliation{Center for Computational Natural Sciences and Bioinformatics, International Institute of Information Technology, Hyderabad, Gachibowli, Telangana 500032, India}

\author{Swati Choudhary}
\affiliation{Centre for Quantum Science and Technology, International Institute of Information Technology, Hyderabad, Gachibowli, Telangana 500032, India}
\affiliation{Center for Computational Natural Sciences and Bioinformatics, International Institute of Information Technology, Hyderabad, Gachibowli, Telangana 500032, India}
\affiliation{Harish-Chandra Research Institute,  A CI of Homi Bhabha National Institute, Jhunsi, Uttar Pradesh  211 019, India}

\author{Vishal Anand}
\affiliation{Centre for Quantum Science and Technology, International Institute of Information Technology, Hyderabad, Gachibowli, Telangana 500032, India}
\affiliation{Center for Computational Natural Sciences and Bioinformatics, International Institute of Information Technology, Hyderabad, Gachibowli, Telangana 500032, India}

\author{Siddhartha Das}
\email{das.seed@iiit.ac.in}
\affiliation{Center for Security, Theory and Algorithmic Research, International Institute of Information Technology, Hyderabad, Gachibowli, Telangana 500032, India}
\affiliation{Centre for Quantum Science and Technology, International Institute of Information Technology, Hyderabad, Gachibowli, Telangana 500032, India}

\begin{abstract}
The erasure of information is fundamentally an irreversible logical operation, carrying profound consequences for the energetics of computation and information processing. We investigate the thermodynamic costs associated with erasing (and preparing) quantum processes. Specifically, we analyze an arbitrary bipartite unitary gate acting on logical and ancillary input-output systems, where the ancillary input is always initialized in the ground state. We focus on the adversarial erasure cost of the reduced dynamics --- that is, the minimal thermodynamic work cost to erase the logical output of the gate for any logical input, assuming full access to the ancilla but no access to any purifying reference of the logical input state. We determine that this adversarial erasure cost is directly proportional to the negative min-entropy of the reduced dynamics, thereby giving the dynamical min-entropy a clear operational meaning. The dynamical min-entropy can take positive and negative values, depending on the underlying quantum dynamics. The negative value of the erasure cost implies that the extraction of thermodynamic work is possible instead of its consumption during the process. A key foundation of this result is the quantum process decoupling theorem, which quantitatively relates the decoupling ability of a process with its min-entropy. This insight bridges thermodynamics, information theory, and the fundamental limits of quantum computation.
\end{abstract}

\maketitle
\section{Introduction}
Harnessing the quantum properties of physical systems at the microscopic scale is crucial for advancing technology and driving the miniaturization of computational processors. As quantum systems become integral to information processing, understanding their energetics becomes essential for the architecture and operation of efficient quantum processors~\cite{Per85,PSE96,VWI09,Sag12,KBL+19,Auf22,DS25}. In quantum computation and information tasks, circuits operate on input quantum states to produce outputs that depend on the structure and function of the circuit. Quantum algorithms often require systems to begin in fixed, predefined pure states~\cite{Kit97,SV99}. Due to the limited availability of quantum resources, it is essential to reset these output states for reuse in subsequent algorithms or protocols. Resetting of quantum states to a desired pure state is erasure of information~\cite{Shi95, Ben03}, which incurs thermodynamic cost. To assess the energetic efficiency of quantum computational devices, it is then crucial to understand the thermodynamic work required to erase and prepare the output of an arbitrary quantum process.

 \begin{figure}
  \centering
  \includegraphics[width=\linewidth]{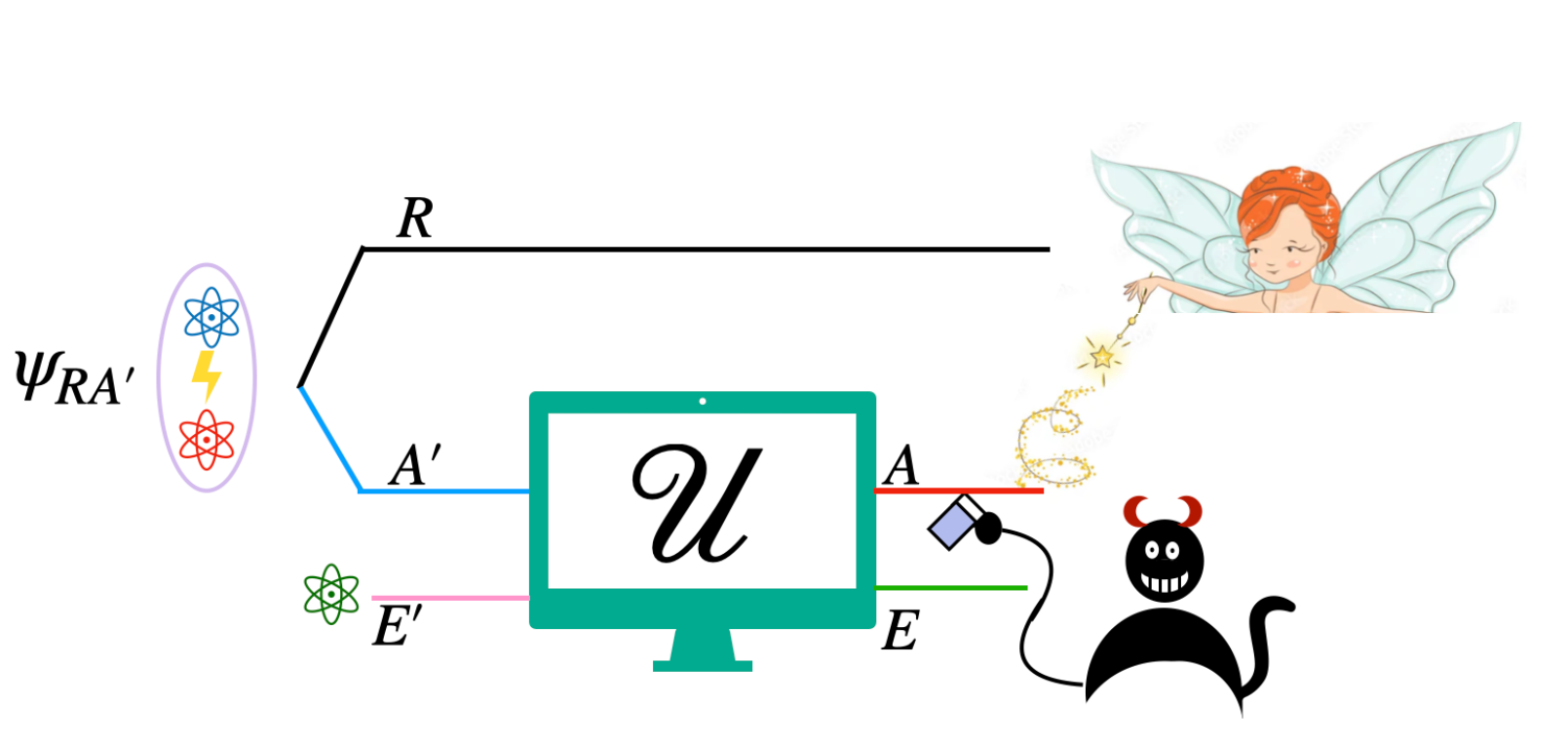}
  \caption{Pictorial representation of the main problem: We consider a bipartite unitary gate (channel) $\mathcal{U}_{A'E'\to AE}$ where $E',E$ are ancilla ($\op{0}_{E'}$ being the initial state). $R$ is a reference that purifies the logical input $A'$. The Hamiltonian of the logical output $A$ is trivial. The goals are to determine: (a) the optimal erasure cost of $A$ of the quantum process $\mathcal{U}$ when the eraser has access to $E$ but not $R$, (b) the optimal preparation cost of $A$ when the preparer has access to $R$ but not ancilla.}
  \label{fig:problem}
\end{figure}
In a seminal work~\cite{Lan61}, Landauer demonstrated that any logically irreversible operation is inherently thermodynamically irreversible, and must be accompanied by a minimum heat dissipation of $k_B T \ln 2$ joules per bit of information~\cite{Ben73, Ben03,BAC+12}. Erasing a single bit of information necessarily induces an entropy change $\Delta S=k_B \ln2$, resulting in heat dissipation $Q=T \Delta S=k_B T \ln2$ into an environment at a temperature $T$. Equivalently, this implies that at least $k_B T \ln2$ of work must be performed on a system to erase one bit of information stored in it. This result, known as Landauer’s principle, forms a foundational correspondence between information theory and thermodynamics~\cite{CGG+16,MSBA16,AH19,SW21,P23}. It also provides a resolution to the famous paradox of Maxwell’s demon. In this, although the demon appears to violate the second law of thermodynamics by using information stored in a finite-size memory register to extract work in a cyclic process, the paradox is resolved when one accounts for the thermodynamic cost of resetting the memory~\cite{Ben03,RW14}. The second law remains intact when the memory register, system, and environment are considered together. Interestingly, in a composite bipartite system, one may aim to erase information from a specific subsystem while leaving the other unaffected. In such scenarios, the thermodynamic work cost of erasure can be significantly reduced by leveraging correlations between the subsystems, particularly when side information or memory in the second subsystem is accessible~\cite{RAR+11, FDOR15, BRLW17}. Under these conditions, the heat dissipated during erasure becomes proportional to the conditional entropy of the subsystem being reset, which is lower than its entropy in general.

Quantum computing and communication devices are composed of several quantum gates~\cite{Kit97,GM02,DBWH21}. Realistically, only finite quantum resources, required quantum states and gates, are available for use. It is desirable to initialize the logical systems to a pure state, at the end or possibly at appropriate stages of the algorithm or protocol run, so that they can be reused~\cite{SV99}. Furthermore, in practice, gates are used only a finite number of times during each computational and information processing tasks~\cite{BV93,Tom21,DBWH21,MY25}. These aspects are important to consider for designing energy-efficient quantum algorithms and processors. We analyze the thermodynamic cost of erasure (and preparation) of an arbitrary quantum gate when the gate is used only once. We show that the erasure cost of a quantum gate is related to its min-entropy. Our framework provides a quantitative method to assess different quantum processors on the basis of the erasure (or preparation) costs of the gate components.

\textit{Problem setup}.--- Our framework centers on bipartite unitary transformations, which induce closed-system evolution in bipartite quantum systems. For an arbitrary bipartite unitary process $\mathcal{U}_{A'E'\to AE}$, we let $A',A$  be logical and $E',E$ be ancillary. We know that the reduced dynamics after locally tracing out $E$ for any state $\omega_{E'}$ uncorrelated with $A'$, i.e., effective process $\mathcal{N}_{A'\to A}(\cdot):=\tr_E\circ\mathcal{U}_{A'E'\to AE}(~\cdot\otimes\omega_{E'})$, is always a completely positive, trace-preserving linear map, also called a \textit{quantum channel}~\cite{Sti55,Lin75}. The ancilla is also called the environment, and without loss of generality, we assume $\omega_{E'}$ to be in the pure, ground state $\op{0}_{E'}$. There always exists an isometric operation $\mathcal{V}^{\mathcal{N}}_{A'\to AE}$ , which is called an isometric extension of a quantum channel $\mathcal{N}_{A'\to A}$ formed by an isometry operator $V_{A'\to AE}=U_{A'E'\to AE}\ket{0}_{E'}$.
\begin{align} 
\mathcal{N}_{A'\to A}(\cdot)&=\tr_E\circ\mathcal{U}_{A'E'\to AE}(~\cdot\otimes\op{0}_{E'}) \label{eq:stines}\\
&=\tr_E(\mathcal{V}^{\mathcal{N}}_{A'\to AE}(\cdot)),\label{eq:iso}
\end{align}
where $\mathcal{V}^\mathcal{N}_{A'\to AE}(\cdot)=V_{A'\to AE}(\cdot)(V_{A'\to AE})^\dag$. 

In this work, we assume that the Hamiltonian of the logical output is trivial. We focus on determining the optimal amount of thermodynamic work required by an eraser to erase the output of a quantum channel for all possible inputs when access to the environment (ancilla) is available and the channel is used only once. We do not allow the eraser to have access to the purifying reference $R$ of the logical input. We refer to this optimal amount of thermodynamic work as the (one-shot) adversarial erasure cost of a given quantum channel $\n_{A'\to A}$. This is exactly the same as the erasure cost of the logical output $A$ from a single use of the quantum circuit $\mathcal{U}_{A'E'\to AE}$~\eqref{eq:stines} when the eraser has access to ancilla $E$ but no access to $R$ (see Fig.~\ref{fig:problem}). We also inspect the preparation cost of a channel, where the objective is to prepare the logical output $A$ when the preparer has access to the reference but no access to the ancilla. The preparation of the channel $\mathcal{N}_{A'\to A}$ requires the preparer to be able to prepare joint reference-logical output state $\mathcal{N}(\psi_{RA'})$ of each possible reference-logical input state $\psi_{RA'}$, from the erased logical system $A$ when $R$ is accessible (see Section~\ref{sec:ap-res} for formal description).

\begin{figure}
  \centering
  \includegraphics[width=\linewidth]{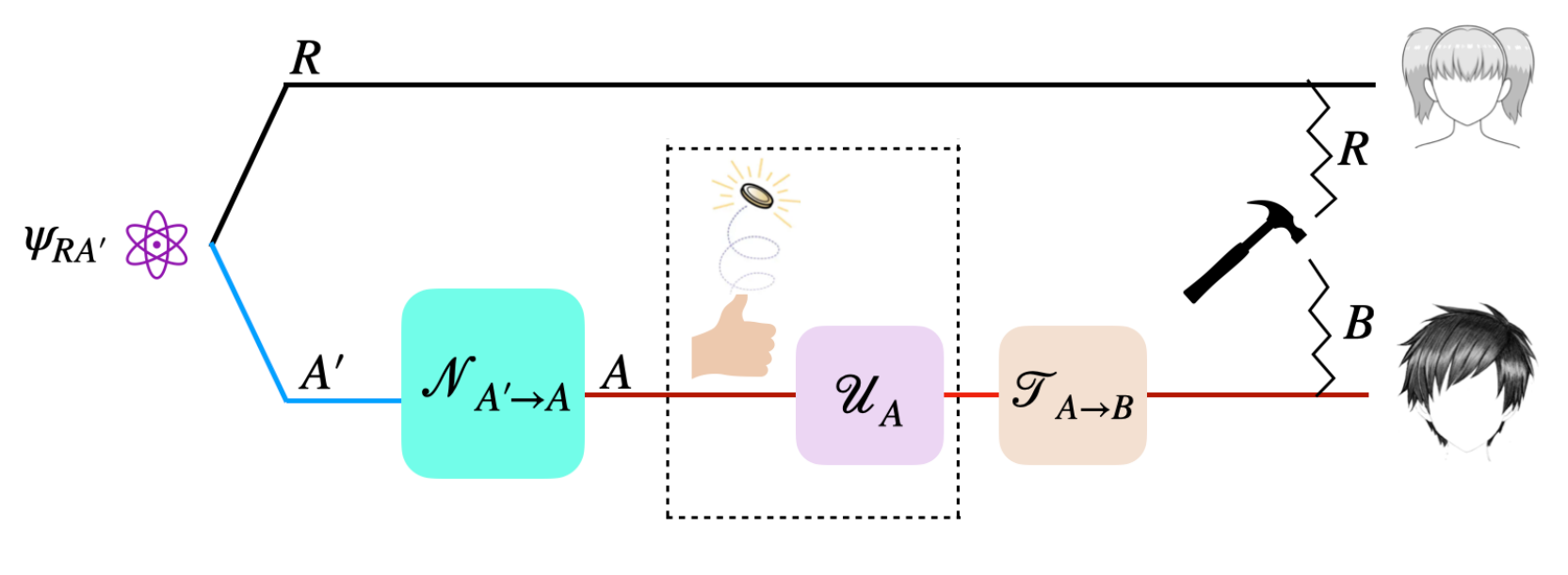}
  \caption{The picture illustrates the mechanism of the decoupling theorem for quantum channels. Decoupling theorem provides a fundamental limit on the ability of a channel $\n_{A'\to A}$ post-processed with $\mathcal{T}_{A\to B}\circ\mathcal{U}_A$, where $\mathcal{U}_A$ is a Haar-random unitary channel, to decouple the output $B$ from its reference $R$. See Theorem~\ref{thm:de-ch} for the formal, precise statement.}
  \label{fig:decoup}
\end{figure}

Main results.--- We provide an operational meaning to the dynamical min-entropy, i.e., min-entropy of a quantum channel~\cite{GW21}, in the context of quantum thermodynamics. We find that the one-shot adversarial erasure cost of a quantum channel, using a reservoir at a fixed temperature $T$, is (approximately) proportional to the negative of its dynamical min-entropy. See Theorem~\ref{thm:workcost} and Proposition~\ref{prop:eras_prep_bounds} for formal, precise statements. We consider the adversarial erasure costs under two different frameworks: (i) a thermodynamic approach~\cite{RAR+11}, and (ii) a resource-theoretic approach (in particular, resource theory of conditional nonuniformity)~\cite{JGW25}. A pivotal underlying concept behind the operational meaning of the dynamical min-entropy is the decoupling theorem for quantum channels, which shows that a channel’s ability to decouple the output from its reference is essentially related to its min-entropy. 

We notice that the adversarial erasure of the output of a channel subsumes the decoupling of the output from the reference as well as the ancilla. This observation provides an insight to the emergence of the dynamical min-entropy as an optimal rate for the one-shot erasure cost under both thermodynamic and resource-theoretic approaches. Bounding the adversarial erasure cost of the logical output of a quantum circuit (gate), in terms of the min-entropy of the reduced dynamics (a quantum channel), allows for quantitative assessment of the energetics of quantum processing and computing devices for the resetting (initializing the states to $\op{0}$) and reuse of logical quantum systems. The adversarial erasure cost of a quantum channel is negative (nonpositive) if it is a PPT channel, i.e., its Choi state remains positive under partial transposition. Entanglement-breaking channels (e.g., measurement, replacer channels) are a kind of PPT channels~\cite{HSR03,HRF20}; though not all PPT channels are entanglement-breaking, all PPT channels have zero quantum capacity under local operations and classical communication~\cite{EVWW01,SY08}. Negative value of the erasure cost implies that we can produce thermodynamic work during the erasure process instead of requiring (consuming) it.

Organization.--- We briefly introduce standard notations, definitions, and facts useful for the derivation of our results and discussion surrounding methods and examples in Section~\ref{sec:prem}. In Section~\ref{sec:er-de}, we discuss the connection between the erasure of a part of a bipartite quantum state and the decoupling of a composite system. We derive the decoupling theorem for quantum channels and dual expressions of the dynamical min-entropy. This provides an intuitive build-up for determining the operational meaning of the dynamical entropy in quantum thermodynamics. We formally introduce the adversarial erasure cost of a quantum process in Section~\ref{sec:er-channel}. In Section~\ref{sec:ap-dec}, we utilize the framework using decoupling of \cite{RAR+11} to define the one-shot adversarial erasure cost of a quantum channel and determine that the erasure cost is approximately upper bounded by the negative of its dynamical min-entropy with high probability. In Section~\ref{sec:ap-res}, we consider the resource theory of conditional nonuniformity~\cite{JGW25} to define the adversarial erasure cost and the preparation cost of a quantum channel. The preparation task is the opposite of the erasure cost. In the preparation task for states, the goal is to prepare a composite state $\rho_{AB}$ from $\op{0}_{A_0}\otimes\rho_B$. In the erasure task for states, the goal is to reset $A$ of a composite state $\rho_{AB}$, i.e., $\rho_{AB}\to \op{0}_{A_0}\otimes\rho_B$. These tasks are performed under conditionally uniformity-covariant channels~\cite{JGW25}. We upper bound the erasure and the preparation costs for quantum channels in terms of the dynamical min-entropies. We show that the zero-error costs are directly proportional to the negative of the dynamical entropy when the reservoir is at a fixed temperature. We discuss the connection between some characteristic properties of quantum channels and their dynamical min-entropies in Section~\ref{sec:examples}. Finally, we discuss the summary of our work and provide concluding remarks in Section~\ref{sec:con}. 

\section{Preliminaries}\label{sec:prem}
\textit{Notations}.--- The separable Hilbert space associated with a system $A$ is denoted as $A$ itself, and its dimension as $|A|$. $\mathbbm{1}_A$ is the identity operator and $\pi_A:=\frac{1}{|A|}\mathbbm{1}_A$ is the maximally mixed state of $A$. For a bipartite operator $\rho_{AB}$, $\rho_A=\tr_B(\rho_{AB})$. Let $\St(A)$ denote the set of all states (density operators) of $A$. $\op{\psi}_A$ denotes that a state $\psi_A=\op{\psi}_A$ is pure. Let $\Ch(A',A)$ denote the set of all quantum channels $\mathcal{N}_{A'\to A}$. The Choi operator of a linear map $\n_{A'\to A}$ is $\Gamma^{\n}_{RA}=\id_R\otimes\n_{A'\to A}(\Gamma_{RA'})$ for $R\simeq A'$, where $\id_R$ is the identity supermap with input-output $R$ and $\Gamma_{RA'}:=\sum_{i,j=0}^{d-1}\ket{ii}\bra{jj}_{RA'}$ with $d=\min\{|R|,|A'|\}$. $\Phi_{RA}:=\frac{1}{d}\Gamma_{RA}$, where, $d=\min\{|R|,|A|\}$, is a maximally entangled state and $\Phi^{\mathcal{N}}_{RA}:=\frac{1}{|A|}\Gamma^{\mathcal{N}}_{RA}$ is a state (Choi state of $\mathcal{N}_{A'\to A}$) if $\mathcal{N}_{A'\to A}$ is a quantum channel. For an isometric channel $\mathcal{V}_{A'\to A}$, we have $\mathcal{V}_{A'\to A}(\cdot):=V_{A'\to A}(\cdot)(V_{A'\to A})^\dag$ and $|A'|\leq |A|$. Wherever it is clear from the context, we will abbreviate our notations, e.g., write $\mathcal{N}_{A'\to A}(\psi_{RA'})$ or $\mathcal{N}(\psi_{RA'})$ for $\id_R\otimes\mathcal{N}_{A'\to A}(\psi_{RA'})$. $\log$ is logarithm with base $2$ and $\ln$ is natural logarithm.

$\norm{\rho}_1:= \tr(\sqrt{\rho^\dag\rho})$ is the trace-norm of an operator $\rho$, and $\norm{\n}_{\diamond}:= \sup_{\rho\in{\St(RA')}}\norm{\id_R\otimes\n_{A'\to A}(\rho_{RA'})}_1$ is the diamond norm of a Hermiticity-preserving map $\n_{A'\to A}$, where it suffices to take $R\simeq A'$ and $\rho_{RA'}$ to be pure. The fidelity between $\rho,\sigma\in\St(A)$ is $F(\rho,\sigma):=\norm{\sqrt{\rho}\sqrt{\sigma}}_1^2$ and $\mathrm{P}(\rho,\sigma):=\sqrt{1-F(\rho,\sigma)}$ is the purified distance. 

\textit{Conditional entropy}.--- Conditional entropy $\mathbf{S}(A|B)_{\rho}$ of a quantum state $\rho_{AB}$ is an entropic function $\mathbf{S}$ that quantifies the uncertainty (randomness) of $A$ when $B$ is accessible. There are several entropic functions derived from different families of relative entropies, see Appendix~\ref{app:entropies} for details. We focus on conditional min-entropies derived from the max-relative entropy $D_{\max}$. The max-relative entropy between a state $\rho_A$ and a positive semidefinite operator $\sigma_A$ is~\cite{Dat09} $D_{\max}(\rho\Vert \sigma)
     =\log\inf_{\lambda}\{\lambda:~ \lambda \sigma\ge \rho\}$. The conditional min-entropy $S_{\min}(A|B)_{\rho}$ of a state $\rho_{AB}$ is defined as~\cite{Ren06,Tom21}
\begin{equation}
   S_{\min}(A|B)_{\rho}:=S^\uparrow_{\min}(A|B)_\rho :=-\inf_{\sigma\in\St(B)}D_{\max}(\rho_{AB}\Vert \mathbbm{1}_A\otimes \sigma_B).
\end{equation}
There's a variant of the conditional min-entropy of a state $\rho_{AB}$ defined as $S^\downarrow_{\min}(A|B)_\rho:=-D_{\max}(\rho_{AB}\Vert \mathbbm{1}_A\otimes \rho_B)$. It follows that $S_{\min}(A|B)_{\rho}\geq S^\downarrow_{\min}(A|B)_\rho$.

\textit{Hypothesis testing conditional entropy.}--- The $\varepsilon$-hypothesis testing relative entropy~\cite{BD10,LR12} between a state $\rho_A$ and $\sigma_A\ge 0$ for $\varepsilon\in[0,1]$ is defined as
    \begin{equation}
D^\varepsilon_{H}(\rho\Vert \sigma)= -\log\inf_{\Lambda}\left\{\tr(\Lambda\sigma) :0\leq \Lambda \leq\mathbbm{1}, \tr(\rho \Lambda) \ge 1-\varepsilon\right\}.
    \end{equation}
For $\varepsilon=0$, $D_{H}^0(\rho\Vert\sigma)=-\log\tr(\Pi_\rho\sigma)$ where $\Pi_\rho$ is the projector onto the support of $\rho$. The $\varepsilon$-hypothesis testing conditional entropy of $\rho_{AB}$ is defined as
\begin{align}
      S^{\varepsilon}_{H}(A|B)_\rho &:=-\inf_{\sigma\in\St(B)}D^\varepsilon_{H}(\rho\Vert \mathbbm{1}_A\otimes \sigma_B).
\end{align}

\textit{Dynamical entropy}.--- A quantum channel is a dynamical process, so we also refer to the entropy of a quantum channel as the dynamical entropy. Let $\mathcal{R}^\omega_{A'\to A}$ denote a replacer map that always outputs a fixed operator $\omega_A$, $\mathcal{R}^\omega_{A'\to A}(\rho_{A'}):= \tr(\rho_{A'})\omega_A$ for all $\rho_{A'}$. The min-entropy $S_{\min}[\mathcal{N}]$ of a quantum channel $\n_{A'\to A}$ is defined as~\cite{GW21}
\begin{align}
    S_{\min}[\n] &:=-\sup_{\psi\in\St(RA')}D_{\max}(\n_{A'\to A}(\psi_{RA'})\Vert\mathcal{R}^{\mathbbm{1}}_{A'\to A}(\psi_{RA'})),\nonumber\\
    &= - D_{\max}(\Phi^{\n}_{RA}\Vert\pi_{R}\otimes\mathbbm{1}_A)= S_{\min}^\downarrow(A|R)_{\Phi^\n},
\end{align}
where $\Phi^{\n}_{RA}$ is the Choi state of $\n$, and it is proven that~\cite{GW21}
\begin{align}\label{eq:smin-uparrow}
    S_{\min}[\n]& =\inf_{\op{\psi}\in\St(RA')}S_{\min}^\downarrow(A|R)_{\n(\psi_{RA'})}\nonumber\\
&=\inf_{\op{\psi}\in\St(RA')}S_{\min}^\uparrow(A|R)_{\n(\psi_{RA'})}.
\end{align}

\textit{Smoothened entropies}.---  There are families of entropies parametrized by a smoothing parameter $\varepsilon\in[0,1]$, which appears as an error probability in various information-theoretic tasks (protocols)~\cite{Ren06}. An $\varepsilon$-ball around a state $\rho_A$ is defined as $\mathcal{B}^\varepsilon(\rho_A)=\{\sigma_A:~ \sigma_A\ge 0, \tr(\sigma_A)\le 1, P(\rho_A,\sigma_A)\le \varepsilon\}$. The smoothened conditional min-entropy $S^\varepsilon_{\min}$ and its variant $S^{\downarrow,\varepsilon}_{\min}$ for state $\rho_{AB}$ are defined as~\cite{Ren06}
    \begin{align}
         S_{\min}^{\varepsilon}(A|B)_\rho&:=\sup\limits_{\widetilde{\rho}\in \mathcal{B}^\varepsilon(\rho_{AB})}  S_{\min}(A|B)_{\widetilde{\rho}}~,\\
     S_{\min}^{\downarrow,\varepsilon}(A|B)_\rho&:=    \sup\limits_{\widetilde{\rho}\in \mathcal{B}^\varepsilon(\rho_{AB})}  -D_{\max}(\widetilde{\rho}_{AB}\Vert\mathbbm{1}_A\otimes\tilde{\rho}_B).
    \end{align}
An $\varepsilon$-ball around a channel $\n_{A'\to A}$ is $\mathcal{B}^\varepsilon[\n]=\{\m_{A'\to A}:~ \mathrm{P}[\n,\m]\le \varepsilon,~\m\in\mathrm{Ch}(A',A)\}$, where \begin{equation}
    \mathrm{P}[\n,\m]:=\sup_{\psi\in\St{(RA')}}\mathrm{P}(\n(\psi_{RA'}),\m(\psi_{RA'}))
\end{equation}
is the purified distance between the two channels $\n,\m\in\Ch(A',A)$ and it suffices to take the supremum over pure states $\psi_{RA'}$ such that $R'\simeq A$. Using this, we define the smoothened min-entropy of a channel $\n_{A'\to A}$ as~\cite{GW21} 
\begin{equation}
    S^\varepsilon_{\min}[\n]=\sup_{\m\in B^\varepsilon[\n]}S_{\min}[\m].
\end{equation}

\section{Erasure and Decoupling}\label{sec:er-de}
Erasure of quantum information (or a system) is to transform the quantum state of a given system to a known pure state. If we consider an arbitrary state $\rho_{AB}$, then the erasure of the state of $A$ would involve transforming the state $\rho_{AB}$ to $\op{0}_A\otimes\rho_B$ which decouples $A$ from $B$ even if $AB$ was initially correlated. In general, for an arbitrary state $\rho_{AB}$, local operations on $A$ is said to decouple it from $B$ if the final state is a product state $\omega_A\otimes\rho_B$ for some state $\omega_A$ while the local (marginal) state of $B$ remains intact. In this sense, erasure is thought to be a particular instance of decoupling where $\omega_A$ has to be a known pure state. An alternate, crude approach could be to think the task of erasing $A$ to include decoupling as a subroutine where $A$ is first brought in contact of a thermal reservoir~(cf.~\cite{Bar00,MDP22}). $A$ gets thermalized and is decoupled from $B$, due to this process $\rho_{AB} \to\gamma^\beta_A\otimes\rho_B$ for a thermal state $\gamma^\beta_A$, where $\beta:=(k_BT)^{-1}$ denotes the inverse temperature of the reservoir. By tuning certain Hamiltonian parameters (associated with the system and reservoir), the state $\gamma^\beta_A$ is transformed approximately to $\op{0}_A$.

The Landauer's principle states that the total (minimum) amount of work (on an average~\footnote{Here, work cost is estimated per copy when asymptotically many copies ($\rho^{\otimes n}$ for $n\to\infty$) of the state $\rho$ is available.}) needed to erase the state $\rho_A$ in contact with a bath at temperature $T$ is lower bounded by $S(A)_{\rho}k_BT\ln 2$, where $S(A)_{\rho}:=S(\rho_A):=-\tr(\rho\log\rho)$ is the von Neumann entropy of $\rho_A$; it holds that $\lim_{n\to \infty}\frac{1}{n}S^\varepsilon_{\min}(\rho^{\otimes n})=S(\rho)$~\cite{Ren06}. The procedure of the erasure of a qubit can be described as follows: Assume a two-level system, with energy eigenstates $\ket{\uparrow}$ and $\ket{\downarrow}$, in a completely mixed state $\pi$. We would like to transform $\pi$ to the state $\ket{\downarrow}$. We couple the system $\pi$ to a bath at temperature $T$ and then manipulate the energy gap $\Delta E$ between the energy levels such that the occupation probability $p_{\downarrow}=(1+e^{-\Delta E/{k_BT}})^{-1}$ of energy level $\ket{\downarrow}$ approaches 1. The process of changing the energy gap between the levels to infinity through an isothermal process requires $k_BT\ln 2$ joules of work.

We now discuss the decoupling theorem which provides insight into the erasure cost of the state of a single copy of a system when side information is available~\cite{RAR+11}. We elaborate on the method and reasoning for the reduction in work cost when compared to the crude (na\"ive) procedure in the next section.

\textit{Decoupling theorem for states}~\cite{DBWR14} (cf.~\cite{MBD+17}): Let $\varphi_{RA}$ be a state, $\varepsilon\in(0,1)$, and $\mathcal{T}_{A\to B}$ is a completely positive map such that its Choi operator $\Gamma^{\mathcal{T}}_{AB}$ satisfies $\tr(\Gamma^{\mathcal{T}}_{AB})\leq |A|$. Under the action of local unitary operators $U_A$ chosen uniformly (with respect to Haar measure over the full unitary group $\mathbb{U}$ on $A$) and followed by $\mathcal{T}_{A\to B}$, we have
\begin{align}
        \int_{\mathbb{U}(A)}\norm{\mathcal{T}_{A\to B}\circ\U_A(\varphi_{RA})-\varphi_R\otimes\Phi^\mathcal{T}_{B}}_1 \d\!U
        \le 2^{-\frac{1}{2} (S^{\varepsilon}_{\min}(A|R)_{\varphi} + S^{\varepsilon}_{\min}(A|B)_{\Phi^\mathcal{T}})} + 12 \varepsilon, \label{eq:decoup}
    \end{align}
where $\U_A(\cdot):= U_A(\cdot)U^\dag_A$ and $\Phi^\mathcal{T}_B=\frac{1}{|A|}\tr_A(\Gamma^\mathcal{T}_{AB})$. $\mathcal{T}$ is a completely positive, trace subpreserving map, e.g., partial trace, measurement operations, etc. By performing a local action $\mathcal{T}_{A\to B}\circ\mathcal{U}_A$ on $A$ of $\varphi_{RA}$, on an average over the choices of $U_A$, $B$ can be decoupled approximately (up to an error $\varepsilon$) from $R$ if 
\begin{equation}
    S^{\varepsilon}_{\min}(A|R)_{\varphi} + S^{\varepsilon}_{\min}(A|B)_{\Phi^\mathcal{T}}\gtrapprox 0.
\end{equation}
The above discussions suggest that the adversarial erasure cost of a quantum channel is delimited by its decoupling capability. A quantum channel $\mathcal{N}_{A'\to A}$ is a good decoupler if we can decouple $A$ from $R$ considerably well for all input states $\rho_{RA'}$. To quantitatively analyze the decoupling capability of a channel, we derive the channel version of the decoupling theorem (see Fig.~\ref{fig:decoup}). We make use of Eq.~\eqref{eq:decoup} and the following lemma to derive the theorem. See Appendices~\ref{lem:smooth-entropy} and \ref{app:proof_de_ch} for the detailed proof of the lemma and theorem below, respectively.
\begin{lemma}\label{lem:smooth-ch}
    For a quantum channel $\n_{A'\to A}$, we have
    \begin{equation}
      S^\varepsilon_{\min}[\n]\leq  \inf_{\op{\psi}\in\St(RA')}S^{\varepsilon}_{\min}(A|R)_{\mathcal{N}(\psi)}.
    \end{equation}
\end{lemma}
\begin{theorem}[Decoupling theorem for processes]\label{thm:de-ch}
    Let $\mathcal{N}_{A'\to A}$ be a quantum channel, $\mathcal{T}_{A\to B}$ a completely positive map such that $\tr(\Gamma^\mathcal{T}_{AB})\leq |A|$, and $\varepsilon\in(0,1)$. The distance of the channel $\mathcal{N}$ post-processed by $\mathcal{T}\circ\U_A$, when $U_A$ is chosen uniformly at random from the Haar measure over the full unitary group $\mathbb{U}$ on $A$, with the uniformly randomizing channel $\mathcal{R}^{\pi}$ post-processed by $\mathcal{T}$ is upper bounded as
    \begin{align}
  \forall~\op{\psi}\in\St(RA'),\quad      \int_{\mathbb{U}(A)}\norm{\mathcal{T}_{A\to B}\circ\mathcal{U}_A\circ\n_{A'\to A}(\psi_{RA'})-\mathcal{T}_{A\to B}\circ\mathcal{R}^{\pi}_{A'\to A}(\psi_{RA'})}_1 \d\! U
        \le 2^{-\frac{1}{2}\left(S^\varepsilon_{\min}[\n]+S_{\min}^{\varepsilon}(A|B)_{\Phi^\mathcal{T}}\right)}+12\varepsilon,
    \end{align}
where $\Phi^\mathcal{T}_{AB}:=\frac{1}{|A|}\Gamma^{\mathcal{T}}_{AB}$ is a scaled Choi operator of $\mathcal{T}_{A\to B}$, $\mathcal{R}^\pi_{A'\to A}(\rho_{A'})=\pi_A \forall \rho\in\St(A')$, $\mathcal{U}_A(\cdot):=U_A(\cdot)U^\dag_A$.
\end{theorem}
$\mathcal{R}^\pi_{A'\to A}$ is the uniformly mixing channel, also called completely depolarizing channel, that outputs the maximally mixed state $\pi_{A}=\frac{1}{|A|}\mathbbm{1}_A$, irrespective of what the input state $\rho_{A'}$ is. We can always input $A'$ of an arbitrary bipartite state $\varphi_{RA'}$ to the quantum channel $\mathcal{N}_{A'\to A}$, then $R$ is called a reference system (to $\n$). We can also purify an arbitrary input state $\rho_{A'}$ to $\psi^{\rho}_{RA'}$. Let $\mathcal{T}_{A\to B}$ be a quantum channel, then $\mathcal{T}\circ\mathcal{U}_A$ is also a channel. Theorem~\ref{thm:de-ch} implies that for a quantum channel $\n_{A'\to A}$, postprocessing of its output $A$ by $\mathcal{T}_{A\to B}\circ\mathcal{U}_A$ is decoupled approximately (up to an error $\varepsilon$) from the reference $R$, on average over choices of $U_A$, if
\begin{equation}
     S^\varepsilon_{\min}[\n]+S_{\min}^{\varepsilon}(A|B)_{\Phi^\mathcal{T}}\gtrapprox 0.
\end{equation}

The intrinsic connection between decoupling and the dynamical min-entropy is also suggested by the following \textit{dual} expressions, see Appendix~\ref{app:proof_channel_entropy} for the proof.
\begin{proposition}\label{theo:channel_entropy}
The min-entropy $S_{\min}[\n]$ of a quantum channel $\n_{A'\to A}$ is
\begin{align}
     S_{\min}[\n]
    &=-\sup_{\op{\psi}_{RA'}}\sup_{\m\in\Ch(R,\bar{A})} \log(|A|F(\m\otimes\n(\psi_{RA'}),\Phi_{A\bar{A}}))\label{eq:min-sing}\\
& =-\sup_{\rho\in\St(A')}\sup_{\sigma\in\St(E)}\log (|A|F(\mathcal{V}^\n_{A'\to AE}(\rho_{A'}),\pi_{A}\otimes\sigma_E)),\label{eq:min-dec}
\end{align}
$\mathcal{V}^\n_{A'\to AE}$ is an isometric extension channel  of $\n_{A'\to A}$.
\end{proposition}
The above proposition illustrates that the dynamical min-entropy is associated with the ability of the channel to preserve the singlet (a maximally entangled state) up to a local operation~\eqref{eq:min-sing} and the ability of its isometric extension channel to keep its output decoupled from the environment (ancilla)~\eqref{eq:min-dec}. The dynamical min-entropy is bounded as $-\log\min\{|A'|,|A|\}\leq S_{\min}[\n]\leq \log|A|$ for any quantum channel $\n_{A'\to A}$~\cite{GW21,KRS09}. Moreover, it is uniformly continuous on the space of quantum channels, as we state in the following lemma.
\begin{lemma}\label{lem:continuity}
    The dynamical min-entropy is continuous. For any two quantum channels $\n_{A'\to A}$ and $\m_{A'\to A}$ that are $\delta$-close, $\frac{1}{2}\Vert\n-\m\Vert_\diamond\le \delta$, we have   
    \begin{equation}
        \abs{S_{\min}[\n]-S_{\min}[\m]}\le \frac{1}{\ln 2}|A|\min\{|A|,|A'|\}\delta.
    \end{equation}
\end{lemma}
A proof is provided in Appendix~\ref{app:continuity_proof}. Moreover, the dynamical min-entropy is monotonically nondecreasing under the action of an $\mathcal{R}^{\mathbbm{1}}$-subpreserving map~(see Appendix~\ref{app:monotone_minentropy} for the statement and proof). Both these properties are desirable properties of an entropy function of channels~\cite{Gou19}.

\section{Erasure and preparation costs of a channel}\label{sec:er-channel}
In this section, we estimate the one-shot preparation and adversarial erasure costs of a channel. To formally determine the optimal rate of erasure of a quantum channel, we adapt and employ two frameworks and protocols for the erasure of a system when side information is provided: (a) the thermodynamic approach introduced in~\cite{RAR+11} and (b) the resource-theoretic approach introduced in~\cite{JGW25}. We determine the optimal cost for the preparation of a channel using resource-theoretic framework in~\cite{JGW25}. The direct relation of the min-entropy of a quantum channel with its decoupling ability provides intuition for the optimal work costs of the adversarial erasure and preparation of quantum channels. 

\subsection{Adversarial erasure via decoupling: a thermodynamic approach}~\label{sec:ap-dec}
We begin by directly utilizing the decoupling approach in~\cite{RAR+11} to determine the adversarial erasure cost of a channel. We will see that if a system is decoupled from a reference, it indicates a stronger correlation between the system and the purifying environment. This correlation with the environment can be used to reduce the erasure cost.

We assume that only a single use of a quantum channel $\n_{A'\to A}$ or its isometric extension $\mathcal{V}^\n_{A'\to AE}$ is allowed. For each pure input state $\psi_{RA'}$ to an isometric extension channel $\mathcal{V}^\n_{A'\to AE}$ of $\n_{A'\to A}$, we have a pure output state $\varphi_{RAE}=\mathcal{V}^\n_{A'\to AE}(\psi_{RA'})$. We can use the decoupling theorem for states to show that there exists a subsystem $A_1$ of $A$ that is $\delta'$-decoupled from $R$, i.e., $\frac{1}{2}\norm{\varphi_{RA_1}-\varphi_{R}\otimes\pi_{A_1}}_1\leq \delta'$. To see this, let the post-processing sub-channel $\mathcal{T}$ in the decoupling theorem be a partial trace channel $\tr_{A_2}$, for $A=A_1\otimes A_2$. The dimension of $A_1$ is bounded from below as~\cite[Supplementary Lemma III.3.]{RAR+11} 
\begin{equation}\label{eq:decoup_size}
    \log|A_1|\ge \frac{1}{2}\left(\log|A|+S_{\min}^{\varepsilon}(A|R)_{\n(\psi)}\right)+\log(2\delta'-12\varepsilon).
\end{equation}
We can find an approximate purification of this subsystem $A_1$, in the space $A_1\otimes A_2\otimes E$, that is $\sqrt{2\delta'}$-close to a maximally entangled state. This purified state can be used to extract $2\log|A_1|k_BT\ln 2$ amount of work by attaching the whole system $A_1A_2E$ with a thermal bath at a temperature $T$, see Appendix~\ref{app:protocol} for details. The correlations with $E$ can be used to reduce uncertainty about $A$, and thereby reduce the work cost to erase $A$. The uncertainty about $A$ conditioned on $E$ decreases if the system $A$ is strongly correlated with $E$, which would imply strong decoupling from $R$ for pure state $\varphi_{RAE}$. This uncertainty is quantified as $\log|A|-2\log|A_1|$. Therefore, the total work cost $W(A|E)_{\varphi}$ of erasing the state of $A$ while utilizing the correlations with the environment $E$ is given by (cf.~\cite{RAR+11})
\begin{equation}\label{eq:rareq}
    W_{\mathrm{eras}}(A|E)_{\mathcal{V}^\n(\psi)}\le(\log|A|-2\log|A_1|)k_BT\ln 2.
\end{equation}
Using the bound on $|A_1|$ from Eq.~\eqref{eq:decoup_size}, we get the work cost required to erase the system $A$ conditioned on system $E$ is bounded as 
\begin{align}
       W_{\mathrm{eras}}(A|E)_{\mathcal{V}^\n(\psi)}
       \le (-S_{\min}^{\varepsilon}(A|R)_{\n(\psi)}-2\log(2\delta'-12\varepsilon))k_BT\ln2. \label{eq:rar-bound}
       \end{align}
The above relation tells us that given the side information in the system $E$ about the system $A$, it is cheaper to erase the system $A$, i.e., with lesser work cost, if the system $A$ and $R$ have weaker correlations.

\textit{Adversarial erasure cost of a channel}.--- The one-shot thermodynamic work cost to erase the output $A$ of a quantum channel $\mathcal{N}_{A'\to A}$, when the eraser has access to ancilla $E$~\eqref{eq:iso} but not to purifying reference of the logical input, is defined as
\begin{equation}
    W_{\mathrm{eras}}[A|E]_{\mathcal{N}}:=\sup_{\op{\psi}\in\St(RA')} W_{\mathrm{eras}}(A|E)_{\mathcal{V}^\n(\psi)},\label{eq:work-ch}
\end{equation}
$\mathcal{V}^\n$ is an isometric extension channel  of $\n$. We now state a main result that bounds the adversarial erasure cost of the channel. It quantifies the amount of work required to erase the output of a channel when ancilla is available as side information, provided that the channel is used only once.

\begin{theorem}\label{thm:workcost}
The adversarial erasure cost $W_{\mathrm{eras}}[A|E]_{\mathcal{N}}$ of a quantum channel $\n_{A'\to A}$ is bounded as
\begin{equation}
W_{\mathrm{eras}}[A|E]_{\mathcal{N}} \leq \left(-S^\varepsilon_{\min}[\n]+\Delta\right)k_BT\ln 2, 
\end{equation}
with the probability greater than $1-\delta$, where $\delta:=\sqrt{2^{-\Delta/2}+12\varepsilon}$, for all $\delta,\varepsilon> 0$.
\end{theorem}
To arrive at the above result, we mainly utilize Eq.~\eqref{eq:rar-bound} and Lemma~\ref{lem:smooth-ch}, and other related properties (e.g., duality~\cite{TCR10}) of the conditional min-entropy. See Appendix~\ref{app:proof_workcost} for proof details. We discuss implications of certain properties of the reduced dynamics (channel) on the dynamical min-entropy in Section~\ref{sec:examples} , and therefore would also be consequential for the adversarial work cost.

\subsection{Resource-theoretic preparation and adversarial erasure costs}\label{sec:ap-res}
We begin by briefly reviewing the (one-shot) erasure and preparation costs of a quantum system within the framework of the resource theory of conditional nonuniformity, as introduced in \cite{JGW25}. In this resource theory, a state $\rho_{AB}$ is a free state if $A$ is maximally mixed and uncorrelated with $B$, i.e., $\pi_A\otimes\rho_{B}$, and a bipartite operation $\mathcal{G}_{A'B'\to AB}$ is free if it is conditionally uniformity-covariant, i.e.,
\begin{equation}
    \mathcal{G}_{A'B'\to AB}\circ\mathcal{R}^\pi_{A'\to A'}=\mathcal{R}^\pi_{A\to A}\circ\mathcal{G}_{A'B'\to AB}.
\end{equation}
Systems with the same alphabet are supposed to be with the same party or lab, e.g., $A,A'$ can be considered to be with Gargi and $B,B'$ with Nila. We will use notation $A'A:BB'$ to denote such a bipartition. See~\cite{JGW25} for the detailed discussion on the resource theory.

For a composite system $AB$ in a state $\rho_{AB}$, (i) the (one-shot) erasure cost of $A$ conditioned on $B$ when $A$ is brought in contact of a bath at temperature $T$, up to an error $\mu\in[0,1]$, is defined as~\cite{JGW25} 
\begin{align}
     \widetilde{W}^{\mu}_{\mathrm{eras}}(A|B)_{\rho} :=\inf_{
    A_0,A_1,\mathcal{G}}\left\{\log|A_1|-\log|A_0|;\dfrac{1}{2}\left\Vert \mathcal{G}(\op{0}_{A_1}\otimes \rho_{AB}) - \op{0}_{A_0}\otimes \op{0}_A \right\Vert_1\le \mu\right\}k_BT\ln 2,
 \end{align}
 where $\mathcal{G}_{A_1A:B\to A_0A}$ is a conditionally uniformity-covariant channel, which for zero-error cost is of the form 
\begin{align}
    \mathcal{G}(\sigma)=\tr(\Lambda\sigma)\op{0}_{A_0}\otimes\op{0}_{A}+\tr((\mathbbm{1}_{{A_1}AB}-\Lambda)\sigma)\frac{\mathbbm{1}_{A_0A}-\op{0}_{A_0}\otimes\op{0}_{A}}{|A_0||A|-1},
\end{align}
where $0\le\Lambda\le \mathbbm{1}_{A_1AB}$ and (ii) the (one-shot) work cost of preparation of $A$ conditioned on $B$ when $A$ is brought in contact of a bath at temperature $T$, up to an error $\mu\in[0,1]$, is defined as~\cite{JGW25} 
\begin{align}
     \widetilde{W}^{\mu}_{\mathrm{prep}}(A|B)_{\rho} :=\inf_{
    A_0,A_1,\mathcal{G}}\left\{\log|A_0|-\log|A_1|;\dfrac{1}{2}\left\Vert \op{0}_{A_1}\otimes\rho_{AB}- \mathcal{G}(\op{0}_{A_0}\otimes \op{0}_A) \right\Vert_1\le \mu\right\}k_BT\ln 2,
 \end{align}
where $\mathcal{G}_{A_0A\to A_1AB}$ is a conditionally uniformity-covariant channel, which for zero-error cost is of the form 
\begin{align}
    \mathcal{G}(\sigma)=\tr(P_0\sigma)\op{0}_{A_1}\otimes\rho_{AB}+\tr((\mathbbm{1}_{A_0A}-P_0)\sigma)\frac{|A_0||A|\mathbbm{\pi}_{A_1A}\otimes\rho_{B}-\op{0}_{A_1}\otimes\rho_{AB}}{|A_0||A|-1},
\end{align}
where $P_0=\op{0}_{A_0}\otimes\op{0}_{A}$.
\color{black}

As discussed earlier, a pure qubit state can be used to extract $k_BT\ln 2$ amount of work by attaching it to a bath at temperature $T$. The one-shot work cost of erasure measures the amount of pure states used to perform the transformation $\op{0}^{\otimes n}\otimes \rho_{AB}\to \op{0}^{\otimes m}\otimes \op{0}_A\otimes \rho_{B}$ , that is $W_{\mathrm{eras}}(A|B)_{\rho}=n-m$, minimized over the set of allowed operations $\mathcal{G}$ that may carry information from $B$ to $A$, but not vice-versa. The preparation of the state is the opposite of the erasure procedure, $\op{0}_A\otimes\rho_{B}\to \rho_{AB}$, and the one-shot work cost of preparation is defined similarly as the change in the number of pure states, but now the allowed operations $\mathcal{G}$ are restricted to semi-causal communications from $A$ to $B$~\cite{JGW25}. For a state $\rho_{AB}$, the one-shot work costs of preparation and erasure of the system $A$ are~\cite{JGW25}.
\begin{align}\label{eq:prep_erase}
    \widetilde{W}^{\mu}_{\mathrm{prep}}(A|B)_\rho&= -S_{\min}^{\downarrow,\mu}(A|B)_\rho k_BT\ln 2,\\
    \widetilde{W}^{\mu}_{\mathrm{eras}}(A|B)_\rho&=S_{{H}}^{\mu}(A|B)_\rho k_BT\ln 2.
\end{align}

We derive a fundamental limitation on the sum of the preparation and erasure costs of a state below, see Appendix~\ref{app:proof_eras_prep} for the proof.
\begin{lemma}\label{prop:eras_prep}
Given a state $\rho_{AB}$, the sum of the work costs of erasing and preparing the system $A$ conditioned on the system $B$, for error $\mu\in(0,1)$, is bounded from below as
\begin{equation}
        \widetilde{W}^{\mu}_{\mathrm{prep}}(A|B)_\rho+ \widetilde{W}^{\mu}_{\mathrm{eras}}(A|B)_\rho\ge \left[\log\left(1-\frac{\mu}{1-\mu^2}\right)-2\right]k_BT\ln 2,
    \end{equation}
and for $\mu=0$, i.e., zero-error erasure and preparation costs, we have
\begin{equation}
      \widetilde{W}^{0}_{\mathrm{prep}}(A|B)_\rho+ \widetilde{W}^{0}_{\mathrm{eras}}(A|B)_\rho\ge 0.
\end{equation}
\end{lemma}

We now introduce the definition and discuss results pertaining to the adversarial erasure cost and the preparation cost of a quantum channel within the framework of the resource theory of conditional nonuniformity.

\textit{Resource-theoretic erasure and preparation costs of a channel}.--- Within the resource theory of conditional nonuniformity, the (one-shot) adversarial erasure cost of a quantum channel $\n_{A'\to A}$ with an error $\mu\in[0,1]$, is defined as
\begin{equation}\label{eq:rt-e-cost}
    \widetilde{W}^{\mu}_{\mathrm{eras}}[A|E]_\n:= \sup_{\rho\in\St(A')}\widetilde{W}^{\mu}_{\mathrm{eras}}(A|E)_{\mathcal{V}^{\n}(\rho)}.
\end{equation}
The eraser has access to the ancilla $E$ but no access to purifying reference $R$ of the logical input. The (one-shot) preparation cost of a quantum channel $\n_{A'\to A}$ with an error $\mu\in[0,1)$, is defined as
\begin{equation}
    \widetilde{W}^{\mu}_{\mathrm{prep}}[A|R]_\n:= \sup_{\op{\psi}\in\St(RA')}\widetilde{W}^{\mu}_{\mathrm{prep}}(A|R)_{{\n}(\psi)}.
\end{equation}

Now we state the bounds on the work costs for erasure and preparation of a quantum channel in terms of the dynamical min-entropy. See Appendix~\ref{app:eras_prep_bounds_proof} for proof.
\begin{proposition}\label{prop:eras_prep_bounds}
Given a quantum channel $\n_{A'\to A}$ with an isometric extension channel  $\mathcal{V}^\n_{A'\to AE}$ and purifying reference $R$ of the logical input, the work costs of preparing and erasing the channel, with a reservoir (bath) at a fixed temperature $T$, are bounded from above for an error $\mu\in(0,1)$ as
\begin{align}
    \widetilde{W}^\mu_{\mathrm{prep}}[A|R]_\n &\le -S^\mu_{\min}[\n]k_BT\ln 2,\\
\widetilde{W}^\mu_{\mathrm{eras}}[A|E]_\n&\le  \left(-S_{\min}[\n]+\log(1-\mu)\right)k_BT\ln 2.
\end{align}
The bounds saturate for the zero-error work costs, i.e., when $\mu=0$, 
\begin{equation}\label{eq:zero_error_eras_prep}
      \widetilde{W}^0_{\mathrm{prep}}[A|R]_\n=\widetilde{W}^0_{\mathrm{eras}}[A|E]_\n= -S_{\min}[\n]k_BT\ln 2.
\end{equation}
\end{proposition}
The identity~\eqref{eq:zero_error_eras_prep} can be understood to follow from Proposition~\ref{theo:channel_entropy} (explainable using the decoupling theorem of a channel). This identity is a channel-analogue of observation made in \cite[Eq.~(75)]{JGW25}. The zero-error preparation and adversarial erasure costs of a quantum channel are continuous. Given any two quantum channels $\n_{A'\to A}$ and $\m_{A'\to A}$ that are $\nu$-close, $\frac{1}{2}\Vert\n-\m\Vert_\diamond\le \nu$, it follows from Lemma~\ref{lem:continuity} that the zero-error adversarial work costs are close as well,
    \begin{equation}
        \abs{\widetilde{W}^0_{\mathrm{eras}}[A|E]_\n-\widetilde{W}^0_{\mathrm{eras}}[A|E]_\m}\le \nu |A|\min\{|A|,|A'|\}k_BT.
    \end{equation}

Consider a smoothened zero-error adversarial erasure cost of a channel, defined as $\widetilde{W}^{0,\varepsilon}_{\mathrm{eras}}[A|E]_\n:=\sup_{\m\in B^\varepsilon[\n]}\widetilde{W}^{0}_{\mathrm{eras}}[A|E]_\m$. The resource-theoretic smoothened zero-error adversarial erasure cost $\widetilde{W}^{0,\varepsilon}_{\mathrm{eras}}[A|E]_\n$ of a quantum channel $\n_{A'\to A}$ is bounded by the thermodynamic adversarial erasure cost $ W[A|E]_{\mathcal{N}}$~\eqref{eq:work-ch},
\begin{equation}\label{eq:recon}
     W[A|E]_{\mathcal{N}}\leq \widetilde{W}^{0,\varepsilon}[A|E]_{\n}+ k_BT\Delta\ln 2,
\end{equation}
with the probability $1-\delta$, where $\delta:=\sqrt{2^{-\Delta/2}+12\varepsilon}$ for all $\delta,\varepsilon>0$. It follows from the observation that $\sup_{\m\in B^\varepsilon[\n]}\widetilde{W}^{0}_{\mathrm{eras}}[A|E]_\m=-S_{\min}^{\varepsilon}[\n]k_BT\ln 2$ and employing Theorem~\ref{thm:workcost}. The inequality~\eqref{eq:recon} provides a quantitative relation between the adversarial erasure costs defined utilizing thermodynamic~\cite{RAR+11} and resource-theoretic~\cite{JGW25} frameworks.

\section{Properties and min-entropies of channels}~\label{sec:examples}
In previous sections, we argued that the decoupling theorem provides insight or justification for the optimal erasure cost being directly related to the dynamical min-entropy of a channel and formally introduced the results. Informally, we have $W_{\rm eras}[A|E]_{\n}\propto_{\approx} -S_{\min}[\n]$ and $\widetilde{W}^0_{\rm eras}[A|E]_{\n}\propto - S_{\min}[\n]$ for a quantum channel $\n_{A'\to A}$, when the reservoir is at a fixed temperature $T$. Now we discuss the dynamical min-entropy for different classes of quantum channels that are of practical interest. We look into some of their properties and implications on the dynamical min-entropy.  

\textit{PPT channels}.-- A quantum channel $\n_{A'\to A}$ is called a PPT channel if $T_{A}\circ\n_{A'\to A}$, where $T_A$ denotes the transposition map on $A$, is a quantum channel~\cite{EVWW01}. A state $\rho_{AB}$ is called a PPT state if it remains positive under partial transposition~\cite{Rai01,EVWW01}, $T_B(\rho_{AB})\geq 0$; for a PPT state $\rho_{AB}$, $F(\rho_{AB},\Phi_{AB})=\tr[\rho_{AB}\Phi_{AB}]\leq {d}^{-1}$ where $d=\min\{|A|,|B|\}$. A quantum channel $\n_{A'\to A}$ is a PPT channel if and only if its Choi state $\Phi^{\n}_{RA}$ is a PPT state. $S_{\min}[\n]\geq 0$ for all PPT channels $\n_{A'\to A}$, and $S_{\min}[\n]<0$ necessarily implies that the channel $\n_{A'\to A}$ is NPT (not PPT). It follows from the fact that $S^{\downarrow}_{\min}(A|B)_{\rho}\geq 0$ for a PPT state $\rho_{AB}$. PPT channels have vanishing quantum capacity~\cite{EVWW01,SY08}. Proposition~\ref{theo:channel_entropy} shows that the dynamical entropy is directly related to the singlet fidelity of a quantum channel and hence consequential for its one-shot quantum capacity~\cite{BD10,KDWW19}.

\textit{Markovian dynamics}.-- For any quantum dynamics $\n$ of the form $\n_{A'\to A}=\n^2_{A_1\to A}\circ\n^1_{A_0\to A_1}\circ\n^0_{A'\to A_0}$~\cite{Lin75,DKSW18,HRF20}, where $\n^2,\n^1, \n^0$ are quantum channels, such that $\mathcal{R}^{\mathbbm{1}}_{A_0\to A_1}\circ\n^0\leq \mathcal{R}^{\mathbbm{1}}_{A'\to A_1}$ and $\n^2\circ\mathcal{R}^{\mathbbm{1}}_{A_0\to A_1}\leq \mathcal{R}^{\mathbbm{1}}_{A_0\to A}$ (see \cite[Definition 5]{SPSD25}), then $S_{\min}[\n]\geq \max\{S_{\min}[\n^2\circ\n^1],S_{\min}[\n^1\circ\n^0]\}\geq \min\{S_{\min}[\n^2\circ\n^1],S_{\min}[\n^1\circ\n^0]\}\geq S_{\min}[\n^1]$. This is a consequence of the monotonicity of the entropy of a quantum channel under the action of an $\mathcal{R}^{\mathbbm{1}}$-subpreserving quantum superchannel, i.e., $S_{\min}[\Theta(\n)]\geq S_{\min}[\n]$ for a quantum channel $\n$ and an $\mathcal{R}^{\mathbbm{1}}$-subpreserving superchannel $\Theta$ (see Lemma~\ref{lem:monotone_minentropy}).

\textit{Replacer channels}.-- A replacer channel $\mathcal{R}^{\omega}_{A'\to A}$ outputs a fixed state irrespective of what the input state is, $\mathcal{R}^{\omega}_{A'\to A}(\rho_{A'})=\omega_A~ \forall \rho\in\St(A')$ for a fixed state $\omega_A$; we have $S_{\min}[\mathcal{R}^{\omega}]=S_{\min}(A)_{\omega}\geq 0$ for all $\omega\in{\St}(A)$ (also see~\cite{DS25}). We can simulate a replacer channel by utilizing the SWAP gate $\mathcal{S}$,
\begin{align}
    \mathcal{R}^{\omega}_{A'\to A}(\cdot)=\tr_E\circ \mathcal{S}_{A'E'\to AE}(~\cdot\otimes\omega_{E'}),
\end{align}
where $\mathcal{S}_{AE'\to AE}(\cdot)={\rm SWAP}(\cdot){\rm SWAP}^\dag$ and ${\rm SWAP}\ket{i}_{A'}\ket{j}_{E'}=\ket{j}_A\ket{i}_E$. Consider $\phi^{\omega}_{E'Q}$ and $\psi^{\rho}_{RA'}$ to be purified states of $\omega_{E'}$ and $\rho_{A'}$, respectively, then $\mathcal{S}_{A'E'\to AE}(\psi^{\rho}_{RA'}\otimes\phi^{\omega}_{E'Q})=\psi^{\rho}_{RE}\otimes\phi^{\omega}_{AQ}$ is a pure state. $S_{\min}[\mathcal{R}^{\omega}]=0$ if and only if $S_{\min}(A)_{\omega}=0$, and $S_{\min}(A)_{\omega}=0$ if and only if $\omega_A$ is a pure state. $S_{\min}[\n]$ of a quantum channel attains the maximum $S_{\min}[\n]=\log|A|$ if and only if the channel is uniformly mixing (depolarizing) $\n_{A'\to A}=\mathcal{R}_{A'\to A}^{\pi}$, $S_{\min}[\mathcal{R}^\pi]=\log|A|$.  A replacer channel is a perfect decoupler in the sense that it absolutely decorrelates the system with its reference while the state of the reference remains locally preserved, cf.~Proposition~\ref{theo:channel_entropy}. The SWAP operation decouples $A'$ from $R$ and if $A$ is not in a pure state, it gets coupled with a part $Q$ of  the environment (ancilla) $EQ$. If $A$ is in a pure state, then it cannot be coupled with any other system (here, environment). In a dual picture, for an isometric extension channel  $\mathcal{V}^{\mathcal{R}^{\omega}}_{A'\to AE''}$ of a replacer channel $\mathcal{R}^{\omega}_{A'\to A}$, $\mathcal{V}^{\mathcal{R}^{\omega}}_{A'\to AE''}(\psi_{RA'})$ is pure if $\psi_{RA'}$ is a pure state; that is, while $A$ gets decoupled from $R$, $A$ and $R$ both get coupled with the environment $E''$ ($EQ$) if $\rho_{A'}$ is not a pure state.

\textit{Measurement}.-- A general measurement of a quantum state is given by a positive operator-valued measure (POVM). A POVM is a set $\{\Lambda^x\}_{x\in\mathcal{X}}$ such that $\sum_{x\in\mathcal{X}}\Lambda^x=\mathbbm{1},~ \Lambda^x\geq 0~\forall x\in\mathcal{X}$. A quantum measurement is a kind of quantum channel. We can describe a quantum channel $\mathcal{N}_{A'\to X}$ corresponding to a POVM $\{\Lambda^x\}_{x\in\mathcal{X}}$ on $A'$ as
\begin{equation}
    \mathcal{N}_{A'\to X}(\cdot)=\sum_{x\in\mathcal{X}}\op{x}_X\otimes\tr[\Lambda^x_{A'}(\cdot)],
\end{equation}
where $X$ is a classical register and $\{\ket{x}\}_{x\in\mathcal{X}}$ forms an orthonormal basis. The dynamical min-entropy is given by $S_{\min}[\n]=\inf_{\psi\in{\St}(RA')}S_{\min}(X|R)_{\n(\psi)}$, where $\n_{A'\to X}(\psi_{RA'})=\sum_{x\in\mathcal{X}}p_X(x)\op{x}_{X}\otimes\rho^x_{R}$ is a classical-quantum state for some probability distribution $p_X$, $\rho^x_R=\frac{1}{p_X(x)}\tr[\Lambda^x_{A'}(\psi_{RA'})]$, $p_X(x)=\tr[\Lambda^x(\psi_{A'})]$. Alternatively, the dynamical min-entropy is given by $S_{\min}[\n]=S^{\downarrow}_{\min}(X|R)_{\Phi^{\n}}$, where the Choi state $\Phi^\n_{RX}=\sum_{x\in\mathcal{X}}q_X(x)\sigma^x_{R}\otimes\op{x}_X$ is a quantum-classical state for $q_X(x)=\tr[\Lambda^x\pi_{A'}]=\frac{1}{|A'|}\tr[\Lambda^x]$ and $\sigma^x_R=\frac{1}{q_X(x)}\tr[\Lambda^x_{A'}(\Phi_{RA'})]$. The measurement channel $\n_{A'\to X}$ is an entanglement-breaking channel~\cite{HSR03}, as it breaks entanglement between the reference and the input; that is, the Choi state of a measurement channel is a separable state, a PPT state with no entanglement.

\textit{Unitary gates}.-- Gates are channels, and unitary gates are unitary channels. For any unitary gate $\mathcal{U}_{A'\to A}$, we have $S_{\min}[\mathcal{U}]=-\log|A|$. For an isometric channel $\mathcal{V}_{A'\to A}$, $S_{\min}[\mathcal{V}]=-\log|A'|$. An isometric channel is unitary if the input dimension is equal to its output dimension. In general, $S_{\min}[\mathcal{N}]$ attains the minimum, i.e., $S_{\min}[\mathcal{N}]=-\log\min\{|A'|,|A|\}$, if and only if the channel $\mathcal{N}_{A'\to A}$ is an isometric channel.

To simplify our discussion, let us recall Eq.~\eqref{eq:zero_error_eras_prep}: for a quantum channel $\n_{A'\to A}$, the resource-theoretic zero-error preparation and adversarial erasure costs are proportional to $-S_{\min}[\n]$, to be precise, $\widetilde{W}^0_{\mathrm{prep}}[A|R]_\n=\widetilde{W}^0_{\mathrm{eras}}[A|E]_\n= -S_{\min}[\n]k_BT\ln 2$. $S_{\min}[\n]<0$ necessarily implies that the channel $\n_{A'\to A}$ is a NPT channel, and $\n$ being a PPT channel sufficiently guarantees that $S_{\min}[\n]\geq 0$. The negative value of the adversarial erasure cost implies that the thermodynamic work is extractable during the erasure process instead of requiring work. That is, the thermodynamic work is extractable during the erasure process of a PPT channel. If a PPT channel $\n_{A'\to A}$ is such that $S_{\min}[\n]=0$ then, the thermodynamic work is neither gained nor consumed during the erasure and preparation processes as $\widetilde{W}^0_{\mathrm{prep}}[A|R]_\n=\widetilde{W}^0_{\mathrm{eras}}[A|E]_\n= 0$. The zero-error preparation and adversarial erasure costs obtain (a) maximum if and only if the quantum channel is an isometric channel, i.e., $\widetilde{W}^0_{\mathrm{prep}}[A|R]_\n=\widetilde{W}^0_{\mathrm{eras}}[A|E]_\n= (\log d) k_BT\ln 2$, where $d=\min\{|A'|,|A|\}$, if and only if $\n_{A'\to A}$ is an isometric channel, and (b) minimum if and only if the quantum channel is uniformly mixing, i.e., $\widetilde{W}^0_{\mathrm{prep}}[A|R]_\n=\widetilde{W}^0_{\mathrm{eras}}[A|E]_\n=-(\log |A|)k_BT\ln 2$, if and only if $\n_{A'\to A}=\mathcal{R}_{A'\to A}^\pi$.

\textit{Remark}. We notice that the erasure of a quantum channel $\mathcal{N}_{A'\to A}$ is equivalent to transforming the channel to the replacer channel $\mathcal{R}^{\op{\psi}}_{A'\to A}$ that only outputs a fixed pure state, here $\op{\psi}_A=\op{0}_A$ (see also \cite{BGD25a,BGD25b}).

\begin{figure}[h]
    \centering
    \includegraphics[width=1.0\linewidth]{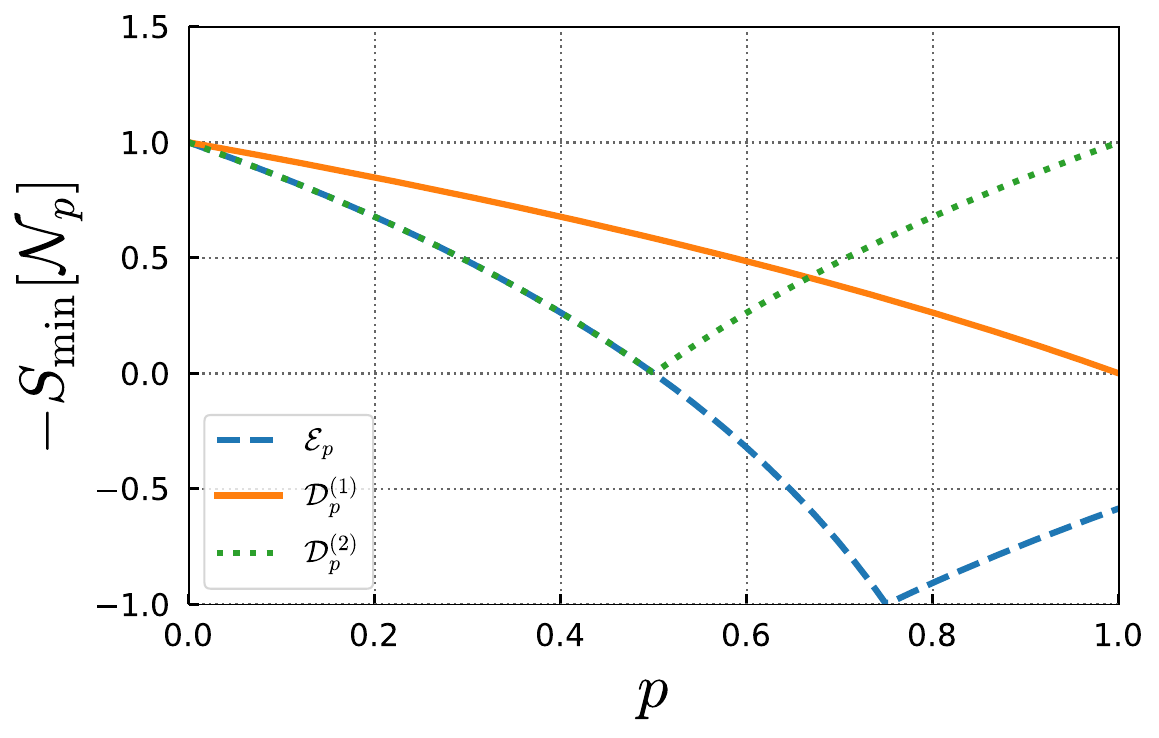}
    \caption{We plot the numerical values for the negative of the dynamical min-entropy $-S_{\min}[\n_p]$ of a $p$-parametrized quantum channel $\mathcal{N}_p$ vs the parameter $p\in[0,1]$, for three families of qubit channels: depolarizing channels $\mathcal{E}_p$~\eqref{eq:polch}, first-kind dephasing channels $\mathcal{D}^{(1)}_p$~\eqref{eq:dep1}, and second-kind dephasing channels $\mathcal{D}^{(2)}_p$~\eqref{eq:dep2}.}
    \label{fig:channel_entropy}
\end{figure}

\subsection{Some numerical examples}
We plot the numerical values of the negative of the dynamical min-entropy for three simple families of qubit quantum channels of interest in Fig.~\ref{fig:channel_entropy}: (a) the depolarizing channel $\mathcal{E}^{A'\to A}_p$ with the probability parameter $p\in[0,1]$,
\begin{equation}  \mathcal{E}_p(\rho):=(1-p)\rho+\frac{p}{3}\sum_{i} \sigma_i\rho\sigma_i, \label{eq:polch}\end{equation}
where $\sigma_i$ are $2\times 2$ Pauli matrices (operators); (b) the first-kind dephasing channel $\mathcal{D}_p^{(1), A'\to A}$ with the probability parameter $p\in[0,1]$,
\begin{equation}   \mathcal{D}^{(1)}_p(\rho):=(1-p)\rho+p~\mathrm{diag}(\rho),\label{eq:dep1} \end{equation} 
where $\mathrm{diag}(\rho)$ denotes the diagonal matrix of $\rho$, $[\mathrm{diag}(\rho)]_{ij}=\delta_{ij}\rho_{ij}$ ($\delta_{ij}$ is the Kronecker delta function); (c) the second-kind dephasing channel $\mathcal{D}_p^{(2), A'\to A}$ with the probability parameter $p\in[0,1]$,
\begin{equation}  \mathcal{D}_p^{(2)}(\rho):=(1-p)\rho+p~ \sigma_3\rho \sigma_3, \label{eq:dep2}\end{equation}

These channels are widely used toy models for open quantum dynamics of qubit systems~\cite{DS05,YE03,Div97,DBWH21}. The depolarizing channel $\n_p=\mathcal{E}_p(\cdot)$ captures symmetric decoherence from disordered media, e.g., randomly oriented birefringent elements. The dephasing channel of the first kind $\n_p=\mathcal{D}_p^{(1)}$ models the decoherence, where, we can study the interaction of the system with the environment using master equations. The dephasing channel of the second kind $\n_p=\mathcal{D}_p^{(2)}$ can be seen as a bit-flip channel in the Hadamard basis~\cite{Div97,YE03}.

 All three channels reduce to the identity channel at $p=0$ and $\mathcal{D}_p^{(2)}$ is a unitary channel also at $p=1$. The zero-error adversarial erasure cost is maximized for unitary channels, since there is no correlation generated between the channel output and the environment. This is reflected in Fig.~\ref{fig:channel_entropy} at $p=0$ for all three channels and also at $p=1$ for $\mathcal{D}_p^{(2)}$, as for a qubit unitary channel $\mathcal{U}_{A'\to A}$, $\widetilde{W}^0_{\rm eras}[A|E]_{\mathcal{U}}=-S_{\min}[\mathcal{U}]=1$. At $p=1$, $\mathcal{D}_p^{(1)}$ is a completely decohering channel and its output is a of the form $\sum_ip_i\op{i}_A$, where $\{p_i\}_i$ is some probability distribution and $\{\ket{i}\}_i$ forms an orthonormal basis. That is, $\mathcal{D}_p^{(1)}$ is an entanglement-breaking channel at $p=1$ and its Choi state $\Phi^{\mathcal{D}^{(1)}_{p=1}}_{RA}=\frac{1}{2}\sum_i\op{i}_R\otimes\op{i}_A$ is a maximally classically correlated state and its purified state $\Phi^{\mathcal{D}^{(1)}_{p=1}}_{RAE}=\frac{1}{2}\sum_{i,j}\ket{iii}\bra{jjj}_{RAE}$, yielding both $\widetilde{W}^0_{\mathrm{prep}}[A|R]_{\mathcal{D}^{(1)}_{p=1}}=\widetilde{W}^0_{\mathrm{eras}}[A|E]_{\mathcal{D}^{(1)}_{p=1}}= 0$. At $p=1/2$, the channel $D_p^{(2)}$ is also a completely decohering channel implying a vanishing adversarial erasure cost.

\section{Discussion}~\label{sec:con}
We find an operational interpretation of the dynamical min-entropy in the context of quantum thermodynamics. We show that the one-shot adversarial erasure cost of a quantum channel is (approximately) proportional to its min-entropy. Adversarial erasure cost is evaluated in a genuinely quantum framework since we assume the evolution of the composite system (system+ancilla) to be closed, i.e., unitary transformation. Our work provides a quantitative way to assess the one-shot erasure and preparation costs of a physical transformation (quantum channel). One-shot costs are important since, in practice, we are only allowed to have a finite use of gates. This makes the framework of determining one-shot costs of erasing and preparing quantum channels instrumental and close to realistic situations.

One of the future directions is to consider the erasure and preparation costs of bipartite quantum channels, when a pair of input and output logical systems are accessible to the observer. It would be interesting to see how these costs would relate to the conditional entropy of a bipartite quantum channel~\cite{DGP24}. The adversarial erasure cost of a quantum channel appears to be related to the quantum capacity and the private randomness capacity of the channel (cf.~\cite{BD10,SD22,DGP24,YHW19}). It could be meaningful to rigorously analyze the relation and explore cryptography or secure communication protocols (cf.~\cite{Ren06,KR11,Das19}) based on the thermodynamic work consumed or gained from the erasure of quantum channels. We believe that the dynamical decoupling theorem, i.e., decoupling theorem for channels, could find some other applications of interest, given that the decoupling theorem for states has found wide applications in quantum science and technology~\cite{DBWR14,MBD+17}.

\begin{acknowledgments}
The authors thank Manabendra Nath Bera for the extensive discussions and useful suggestions on the manuscript. S D thanks Karol Horodecki and Uttam Singh for the discussions. S D and H B thank IIIT Hyderabad for the Seed Grant and the Ministry of Electronics and Information Technology (MeitY), Government of India for the Visvesvaraya Post-Doctoral Fellowship. D G S thanks the Department of Science and Technology, Government of India, for the INSPIRE fellowship. S C acknowledges partial support from the Department of Science and Technology, Government of India, through the QuEST grant with Grant No. DST/ICPS/QUST/Theme-3/2019/120 via I-HUB QTF of IISER Pune, India. S D acknowledges support from the Science and Engineering Research Board (now ANRF), Department of Science and Technology, Government of India, under Grant No. SRG/2023/000217 and the Ministry of Electronics and Information Technology (MeitY), Government of India, under Grant No. 4(3)/2024-
ITEA. S D also acknowledges support from the National Science Centre, Poland, grant Opus 25, 2023/49/B/ST2/02468. V A thanks IIIT Hyderabad for the institutional PhD fellowship. H B thanks the
Institute of Mathematical Sciences (IMSc), Chennai for the hospitality during his visit.
\end{acknowledgments}

\section*{Appendix}
\begin{appendix}
The appendix is organized as follows. In Appendix~\ref{app:prem}, we introduce standard notations, definitions, and facts. We primarily focus on the two well-known families of the relative entropies and conditional entropies derived from them. We review their properties and recall results related to the dynamical entropy that are useful to derive our results. In Appendix~\ref{app:proofs}, we provide detailed proofs of the main results and observations made in the main content. We briefly describe the work extraction protocol in Appendix~\ref{app:protocol}.

\section{Standard notations, definitions, and facts}\label{app:prem}
Let $\St(A)$ and $\Pos(A)$ denote the sets of all quantum states and positive semidefinite operators, respectively, defined on $A$. For an operator $\rho$, $\rho\geq 0$ denotes that $\rho$ is positive semidefinite. An isometry operator $V_{A'\to A}$ satisfies $V^\dag V=\mathbbm{1}_{A'}$ and $VV^\dag=\Pi_{A}$, where $\Pi_A$ is the projector operator, i.e., $\Pi_A^2=\Pi_A$; $|A'|\leq |A|$ for an isometry operator $V_{A'\to A}$, and a unitary operator $V_{A'\to A}$ is an isometry with $|A'|=|A|$.  A quantum channel $\mathcal{N}_{A'\to A}$ is a completely positive, trace-preserving linear map $\mathcal{N}: \St(A')\to \St(A)$. Any quantum channel $\mathcal{N}_{A'\to A}$ can be expressed in terms of Kraus operators $\{K^i_{A'\to A}\}_i$,
\begin{equation}
    \mathcal{N}(\cdot)=\sum_{i}K_i(\cdot)K^\dag_i,
\end{equation}
such that $\sum_iK^\dag_iK_i=\mathbbm{1}_{A'}$. An isometric channel $\mathcal{V}_{A'\to A}$ is a quantum channel with a single Kraus operator $V_{A'\to A}$, and for an isometric channel $|A'|\leq |A|$. A unitary channel $\mathcal{U}_{A'\to A}$ is an isometric channel with the input and output dimensions being the same, i.e., $|A'|=|A|$. A linear map $\mathcal{M}_{A'\to A}$ that doesn't increase the trace, i.e., $\tr(\m(\rho_{A'}))\leq \tr(\rho_{A'})$ for all positive operators $\rho_{A'}$, is called a trace subpreserving map.

A bipartite state $\rho_{AB}$ is called a separable state if it can be expressed as a convex mixture of product states between $A,B$,
\begin{equation}
    \rho_{AB}=\sum_{x}p_X(x){\psi}^x_{A}\otimes{\varphi}^x_B,
\end{equation}
where $p_X(x)\geq 0, \sum_xp_X(x)=1$, and $\psi^x\in\St(A)$ and $\varphi^x\in\St(B)$ for all $x$. If a bipartite state $\rho_{AB}$ is not separable then it is entangled. Let $\Phi_{AB}=\frac{1}{d}\sum_{i,j=0}^{d-1}\ket{ii}\bra{jj}_{AB}$, where $d=\min\{|A|,|B|\}$, i.e., $\Phi_{AB}$ is a maximally entangled state. The Choi state $\Phi^\n_{A'A}$ of a quantum channel $\n_{A'\to A}$ is $\Phi^\n_{A'A}:=\id_{R\to A'}\otimes\n_{A'\to A}(\Phi_{RA'})$.

$k_B$ denotes the Boltzmann constant and the temperature $T$ of a reservoir is an absolute temperature (in units of Kelvin).

\subsection{Review of R\`eyni entropies}\label{app:entropies}
In this appendix, we recall the entropic quantities and their properties from~\cite{P85,P86,Ren06,BD10,WWY14,TBH14,Tom21} necessary to understand the results.

\textit{Sandwiched R\'enyi relative entropies}.-- The family of quantum sandwiched R\'{e}nyi relative entropy between $\rho\in\St(A)$ and $\sigma\in\Pos(A)$ is given by~\cite{MDSFT13,WWY14},
    \begin{align}
        {D}_\alpha(\rho\Vert \sigma) & = \dfrac{1}{\alpha-1}\log \tr\left\{\left( \sigma^{\frac{1-\alpha}{2\alpha}}\rho \sigma^{{\frac{1-\alpha}{2\alpha}}}\right)^\alpha\right\}=\dfrac{1}{\alpha-1}\log \norm{ \sigma^{\frac{1-\alpha}{2\alpha}}\rho \sigma^{{\frac{1-\alpha}{2\alpha}}}}_\alpha^\alpha
    \end{align}
for $\alpha\in(0,1)$ and for $\alpha\in(1,\infty)$ if $\mathrm{supp}(\rho)\subseteq\mathrm{supp}(\sigma)$, else it is set to $+\infty$. For $\alpha\in [\frac{1}{2},1)\cup(1,\infty)$, the sandwiched R\'enyi relative entropy between states is monotonically nonincreasing under the action of quantum channels~\cite{MDSFT13,FL13}. We have, for $\rho,\sigma\in\St(A)$,
\begin{equation}
     {D}_\infty(\rho\Vert \sigma)= \log \norm{\sigma^{-1/2}\rho\sigma^{-1/2} }_\infty
     =\log\inf_{\lambda}\{\lambda:~ \lambda \sigma\ge \rho\},
\end{equation}
\begin{equation}
{D}_{\frac{1}{2}}(\rho\Vert \sigma)=-\log F(\rho,\sigma),
\end{equation}
where $F(\rho,\sigma):=\norm{\sqrt{\rho}\sqrt{\sigma}}_1^2$ is the (Uhlmann) fidelity between quantum states. $D_{\max}(\rho\Vert\sigma):=D_{\infty}(\rho\Vert\sigma)$ is also called the max-relative entropy. We also have $D_\alpha(\rho\Vert\sigma)\leq D_{\beta}(\rho\Vert\sigma)$ for all $\alpha\leq \beta\in(0,1)\cup(1,\infty)$.

\textit{Petz-R\'enyi relative entropy}.-- The family of quantum Petz-R\`enyi relative entropy~\cite{P85,P86} between $\rho\in\St(A)$ and $\sigma\in\Pos(A)$ is given by
    \begin{equation}
        \xoverline{D}_\alpha(\rho\Vert \sigma)= \dfrac{1}{\alpha-1}\log \tr\left\{\rho^\alpha \sigma^{(1-\alpha)}\right\}
    \end{equation}
     for $\alpha\in(0,1)$ and for $\alpha\in(1,\infty)$ if $\mathrm{supp}(\rho)\subseteq\mathrm{supp}(\sigma)$, else it is set to $+\infty$. $\xoverline{D}_\alpha(\rho\Vert \sigma)$ is monotonically nonincreasing under the action of quantum channels for $\alpha\in(0,2]$.

For any $\rho\in\St(A)$ and $\sigma\in\Pos(A)$, we have that both the sandwiched R\'enyi and Petz-R\'enyi relative entropy approaches the (Umegaki) quantum relative entropy,
\begin{equation}
    \lim_{\alpha\to 1}D_{\alpha}(\rho\Vert\sigma)=\lim_{\alpha\to 1}\xoverline{D}_\alpha(\rho\Vert \sigma)=D(\rho\Vert\sigma),
\end{equation}
where $D(\rho\Vert\sigma):=\tr\left[\rho(\log\rho-\log\sigma)\right]$ if $\supp(\rho)\subseteq\supp(\sigma)$ else it is $+\infty$.

We use $\textbf{D}_\alpha$ to denote both the families, sandwiched R\'enyi and Petz-R\'enyi, of quantum relative entropies. Based on the relative entropies, the quantum conditional entropies of a bipartite state $\rho\in \St(AB)$ can be defined as~\cite{Ren06,TBH14}
\begin{align}
    \textbf{S}_\alpha(A|B)_\rho &:=   \textbf{S}^\uparrow_\alpha(A|B)_\rho :=-\inf_{\sigma\in\St(B)}\textbf{D}_\alpha(\rho\Vert \mathbbm{1}_A\otimes \sigma_B),\\
        \textbf{S}^\downarrow_\alpha(A|B)_\rho &:=-\textbf{D}_\alpha(\rho\Vert \mathbbm{1}_A\otimes \rho_B).
\end{align}
The conditional entropy $\textbf{S}(A|B)_{\rho}$ of a state $\rho_{AB}$ quantifies uncertainty in $A$ when side information $B$ is accessible. For a bipartite state $\rho_{AB}$, the von Neumann conditional entropy $S(A|B)_{\rho}=-D(\rho_{AB}\Vert\mathbbm{1}_A\otimes\rho_B)=S(AB)_{\rho}-S(B)_{\rho}$, where $S(A)_{\rho}:=-\tr[\rho\log\rho]$ is the von Neumann entropy of $\rho\in\St(A)$. The conditional min-entropy $S_{\min}(A|B)_{\rho}:=S_{\infty}(A|B)_{\rho}:=\lim_{\alpha\to\infty}S^{\uparrow}
_{\alpha}(A|B)_{\rho}$. 

For a pure quantum state $\psi\in\St(ABC)$, the conditional entropies follow the \textit{duality relations}~\cite{TBH14}
\begin{align}
    {S}_\alpha (A|B)_\psi &= -{S}_{\frac{\alpha}{2\alpha-1}}(A|C)_\psi, ~\text{for}~\alpha\in\left[\frac{1}{2},\infty\right),\\
    {S}_\alpha^\downarrow (A|B)_\psi &= -\xoverline{S}_{\frac{1}{\alpha}}(A|C)_\psi,~\text{for}~\alpha\in(0,\infty),\\
    \xoverline{S}_\alpha^\downarrow (A|B)_\psi &= -\xoverline{S}_{2-\alpha}^\downarrow(A|C)_\psi,~\text{for}~\alpha\in[0,2].
\end{align}

The $\varepsilon$-hypothesis testing relative entropy between $\rho\in\St(A)$ and $\sigma\in\Pos(A)$ is defined as~\cite{BD10,LR12}
    \begin{equation}
D^\varepsilon_{H}(\rho\Vert \sigma)= -\log\inf\{\tr(\Lambda\sigma) :0\leq \Lambda \leq \mathbbm{1}, \tr(\Lambda\rho) \ge 1-\varepsilon\}.
    \end{equation}
For $\varepsilon=0$, $D_{H}^0(\rho\Vert\sigma)=\xoverline{D}_{0}(\rho\Vert \sigma)=-\log\tr(\Pi_\rho\sigma)$, where $\Pi_\rho$ is the projector onto the support of $\rho$.

The $\varepsilon$-hypothesis testing conditional entropy is defined as
\begin{align}
      S^{\varepsilon}_{H}(A|B)_\rho &:=-\inf_{\sigma_B}D^\varepsilon_{H}(\rho\Vert \mathbbm{1}_A\otimes \sigma_B),\\
        S^{\downarrow,\varepsilon}_{H}(A|B)_\rho &:=-D^\varepsilon_{H}(\rho\Vert \mathbbm{1}_A\otimes \rho_B).
\end{align}

\textit{Smoothened entropies}.-- An $\varepsilon$-ball $\mathcal{B}^\varepsilon(\rho)$ around a state $\rho$ is defined as the set of subnormalized states $\sigma$, $\tr(\sigma)\leq 1$ for $\sigma\geq 0$, 
\begin{equation}
    \mathcal{B}^\varepsilon(\rho):=\{\sigma:~ \sigma\ge 0, \tr(\sigma)\le 1, P(\rho,\sigma)\le \varepsilon\},
\end{equation}
where $P(\rho,\sigma):=\sqrt{1-F(\rho,\sigma)}$. We define an $\varepsilon$-ball around a quantum channel $\n_{A'\to A}$ as 
\begin{equation}
    \mathcal{B}^\varepsilon[\n]=\{\m_{A'\to A}:~ P[\n,\m]\le \varepsilon,~\m\in\mathrm{Ch(A',A)}\},
\end{equation}
where $P[\n,\m]:=\sup_{\psi\in\St{RA'}}P(\n(\psi_{RA'}),\m(\psi_{RA'}))$ is the purified distance between two channels $\n_{A'\to A}$ and $\m_{A'\to A}$, and it suffices to take the supremum over pure states.

The smoothened conditional entropies are defined as~\cite{Ren06,BD10} 
\begin{align}
    \textbf{S}_\alpha^{\uparrow,\varepsilon}(A|B)_\rho=\begin{cases}
      \inf\limits_{{\rho}_{AB}\in \mathcal{B}^\varepsilon(\rho)} \textbf{S}_\alpha^{\uparrow}(A|B)_{{\rho}},~ \text{for}~\alpha\in [0,1) \\
      \sup\limits_{{\rho}_{AB}\in \mathcal{B}^\varepsilon(\rho)}  \textbf{S}_\alpha^{\uparrow}(A|B)_{{\rho}},~ \text{for}~\alpha\in (1,\infty)
    \end{cases},
\end{align}
and
\begin{align}
    \textbf{S}_\alpha^{\downarrow,\varepsilon}(A|B)_\rho=\begin{cases}
      \inf\limits_{{\rho}_{AB}\in \mathcal{B}^\varepsilon(\rho)} -\textbf{D}_\alpha({\rho}_{AB}\Vert\mathbbm{1}_A\otimes\rho_B),~ \text{for}~\alpha\in [0,1)  \\
      \sup\limits_{{\rho}_{AB}\in \mathcal{B}^\varepsilon(\rho)}  -\textbf{D}_\alpha({\rho}_{AB}\Vert\mathbbm{1}_A\otimes\rho_B),~ \text{for}~\alpha\in (1,\infty)
    \end{cases}.
\end{align}

For a pure state $\psi\in\St(ABC)$ we have 
\begin{equation}
\begin{aligned}
    {S}_\alpha^{\uparrow,\varepsilon}(A|B)_\psi &= - {S}_\beta^{\uparrow,\varepsilon}(A|C)_\psi, 
   \quad 
    \text{for}~\alpha, \beta \in \left[\dfrac{1}{2}, \infty\right) ~\text{such that} \quad 
    \dfrac{1}{\alpha} + \dfrac{1}{\beta} = 2.
\end{aligned}
\end{equation}

\subsubsection{Relations between the smooth entropies}
Consider $\rho\in\St(A)$ and $\sigma\in\Pos(A)$. For $\varepsilon\in(0,1)$ and $\alpha\in[0,1)\cup(1,2]$, we have \cite{RLD25}
\begin{equation}
  \xoverline{D}_\alpha(\rho\Vert\sigma)-\dfrac{\alpha}{1-\alpha}\log\dfrac{1}{\varepsilon}+\log \dfrac{1}{1-\varepsilon}\le D^\varepsilon_{H}(\rho\Vert \sigma).
\end{equation}
For $\varepsilon\in[0,1)$ and $\alpha\in (\frac{1}{2},1)\cup(1,\infty)$, we have
\begin{equation}
 D^\varepsilon_{H}(\rho\Vert \sigma)\le {D}_{\alpha}(\rho\Vert\sigma)-\dfrac{\alpha}{1-\alpha}\log \dfrac{1}{1-\varepsilon}.
\end{equation}
In particular, following relation is relevant for our work: for $\varepsilon\in[0,1)$,
\begin{equation}
    \xoverline{D}_0(\rho\Vert\sigma)+\log \dfrac{1}{1-\varepsilon}\le D^\varepsilon_{H}(\rho\Vert \sigma)\le {D}_{\infty}(\rho\Vert\sigma)+\log \dfrac{1}{1-\varepsilon}.
    \end{equation}
    
The smooth max-relative entropy is defined as $D^\varepsilon_{\max}(\rho\Vert\sigma):={D}_\infty^\varepsilon(\rho\Vert\sigma):=\inf_{\rho'\in\mathcal{B}^\varepsilon(\rho)}{D}_\infty^\varepsilon(\rho'\Vert\sigma)$. It is related to the hypothesis testing relative entropy in the following ways, see~\cite[Theorem 4]{ABJT19} and \cite{RLD25,DMHB+13,DKF+14}, for $\varepsilon\in(0,1)$ and $\mu\in (0,1-\varepsilon^2)$
\begin{align}
    D_{{H}}^{1-\varepsilon^2-\mu}(\rho\Vert\sigma)&\le {D}_\infty^{\varepsilon}(\rho\Vert\sigma)-\log\frac{\mu^2}{4(1-\varepsilon^2)}~,\\
    D_{{H}}^{1-\varepsilon^2}(\rho\Vert\sigma)&\ge {D}_\infty^{\varepsilon}(\rho\Vert\sigma)-\log\frac{1}{1-\varepsilon^2}.
\end{align}

\subsubsection{Dynamical entropy}
Any desirable dynamical entropy function $\mathbf{S}$ of an arbitrary quantum channel should satisfy following three properties~\cite{Gou19,GW21} (also see~\cite{SPSD25,DGP24}):
\begin{enumerate}
    \item Monotonically nondecreasing under the action of $\mathcal{R}$-preserving superchannels $\Omega$, i.e., $\mathbf{S}[\n]\leq \mathbf{S}[\Omega(\n)]$, for an arbitrary quantum channel $\n$.
    \item Let $\m,\n$ be a pair of quantum channels, then $\mathbf{S}[\n\otimes\m]=\mathbf{S}[\n]+\mathbf{S}[\m]$.
    \item For a replacer channel $\mathcal{R}^{\omega}$, $\mathbf{S}[\mathcal{R}^\omega]=\mathbf{S}(\omega)$.
\end{enumerate}

The sandwiched R\`enyi relative channel entropy~\cite{CMW16} between two quantum channels $\n_{A'\to A}$ and $\m_{A'\to A}$ is given by
\begin{equation}
{D}_\alpha[\n\Vert\m]=\sup_{\psi\in\St{RA'}}{D}_\alpha(\id_R\otimes\n(\psi_{RA'})\Vert\id_R\otimes\m(\psi_{RA'})),
\end{equation}
where it suffices to optimize over pure states $\psi_{RA'}$ and $R\simeq A'$.

The sandwiched R\'enyi entropy of a quantum channel $\n_{A'\to A}$ is defined as~\cite{GW21}, for $\alpha\in(1,\infty)$
    \begin{equation}
        {S}_\alpha[\n]=- {D}_\alpha[\n\Vert \mathcal{R}],
    \end{equation}
where $\mathcal{R}^{\mathbbm{1}}_{A'\to A}$ is a uniformly randomizing map, $\mathcal{R}_{A'\to A}(\sigma_{A'})=\tr(\sigma_{A'})\mathbbm{1}_A$.
We define the smooth channel entropy of a quantum channel as (see \cite{GW21} for $\alpha\to \infty$)
\begin{equation}
    {S}^\varepsilon_\alpha[\n]=\sup_{\m\in B^\varepsilon[\n]}{S}_\alpha[\m]~\text{for}~\alpha\in (1,\infty).
\end{equation}
The entropy of a quantum channel is also called dynamical entropy as a quantum channel is a dynamical process. Properties of the dynamical entropy are discussed in some details in \cite{Yua19,GW21,SPSD25,DGP24,DS25}.

The min-entropy $S_{\min}[\n]:=S_{\infty}[\n]$ of a quantum channel $\n_{A'\to A}$ is~\cite{GW21}
\begin{align}
    S_{\min}[\n]:=\lim_{\alpha\to \infty}S_{\alpha}[\n] &=\inf_{\op{\psi}\in\St(RA')}S_{\min}^\downarrow(A|R)_{\n(\psi)}\\
   &=S_{\min}^\downarrow(A|R)_{\Phi^\n}\\
&=\inf_{\op{\psi}\in\St(RA')}S_{\min}(A|R)_{\n(\psi)}. \label{eq:smin-uparrow-1}
\end{align}

The von Neumann entropy $S[\n]:=\lim_{\alpha\to 1}S_{\alpha}[\n]$ of a quantum channel is~\cite{GW21} (see also \cite[Definition~6]{SPSD25})
\begin{equation}
    S[\mathcal{N}]=-D[\mathcal{N}\Vert\mathcal{R}]=\inf_{\psi\in\St(RA')}S(A|R)_{\n(\psi)},
\end{equation}
where it suffices to optimize over pure states $\psi_{RA'}$.

\textit{Asymptotic equipartition property}.-- For all $\varepsilon\in(0,1)$, the following relations hold~\cite{GW21},
\begin{align}
    \lim_{n\to \infty}\frac{1}{n}S^\varepsilon_{\min}\left[\n^{\otimes n}\right]&\geq S[\n]\\
    \lim_{\varepsilon\to 0}\lim_{n\to \infty}\frac{1}{n}S^\varepsilon_{\min}\left[\n^{\otimes n}\right]&\leq S[\n].
\end{align}
\section{Monotonicity of dynamical min-entropy}\label{app:monotone_minentropy}
We now prove an important monotonicity behavior of the dynamical min-entropy. In particular, we show that the dynamical min-entropy of a quantum channel is nondecreasing under the action of an $\mathcal{R}^{\mathbbm{1}}$-subpreserving superchannel~(cf.~\cite{SPSD25}). The action of a quantum superchannel on a quantum channel can be expressed as concatenation of the channel with preprocessing and postprocessing quantum channels~\cite{Gou19}. Given two uniformly mixing maps $\mathcal{R}^{\mathbbm{1}}_{A\to B}$ and $\mathcal{R}^{\mathbbm{1}}_{C\to D}$, a superchannel $\Theta:\mathrm{Ch}(A,B)\to \mathrm{Ch}(C,D)$ is called $\mathcal{R}^{\mathbbm{1}}$-subpreserving if the linear supermap $\Theta(\mathcal{R}^{\mathbbm{1}}_{A\to B})-\mathcal{R}^{\mathbbm{1}}_{C\to D}$ is completely positive~\cite[Definition 5]{SPSD25}. Under the action of such superchannels, we prove that the dynamical min-entropy is nondecreasing.
\begin{lemma}[cf. Proposition 8 of~\cite{SPSD25}]\label{lem:monotone_minentropy}
Under the action of an $\mathcal{R}^\mathbbm{1}$-subpreserving superchannel $\Theta$,  the dynamical min-entropy of a quantum channel is monotonically nondecreasing:
\begin{equation}
    S_{\min}[\Theta(\n)]\ge     S_{\min}[\n].
\end{equation}
\end{lemma}
\begin{proof}
    Let us first recall the following property of the max-relative entropy between two positive semidefinite operators: For any $\rho\in\St(A)$ and for all $\sigma'\in \mathrm{Pos}(A)$ such that $\sigma'\ge \sigma\in\Pos(A)$, we have 
    \begin{equation}
        D_{\max}(\rho\Vert\sigma')\le         D_{\max}(\rho\Vert\sigma).
    \end{equation}
This can be shown using the definition of the max relative entropy. If $\lambda_0$ is the optimal solution of the semidefinite optimization problem,  $\lambda_0=\inf_{\lambda\ge 0}\{\lambda: \lambda\sigma'\ge \rho\}$, then it is also a feasible solution of the inequality $\lambda_0\sigma\ge \rho$, given $\sigma'\ge\sigma$.  Using this, we see that the max relative entropy of channels, given $\Theta(\mathcal{R}^\mathbbm{1})\le \mathcal{R}^{\mathbbm{1}}$, satisfies 
      \begin{equation}
        D_{\max}[\Theta(\n)\Vert \mathcal{R}]\le         D_{\max}[\Theta(\n)\Vert \Theta(\mathcal{R})].
    \end{equation}
We also know that under any completely positivity preserving (CPP) supermap $\Theta$, the max-relative entropy of channel is monotonically nonincreasing~\cite{Gou19}, therefore
\begin{equation}
    D_{\max}[\Theta(\n)\Vert \mathcal{R}]\le D_{\max}[\Theta(\n)\Vert \Theta(\mathcal{R})] \leq D_{\max}[\n\Vert \mathcal{R}].
\end{equation}
Thus, we conclude that
\begin{align} \label{eq:maxmonotonicity}
      S_{\min}[\Theta(\n)]&\ge S_{\min}[\n].
\end{align}

\end{proof}

\section{Detailed proofs of the results}\label{app:proofs}
\subsection{Proof of Lemma~\ref{lem:smooth-ch}}\label{lem:smooth-entropy}
\begin{nonumberedlemma} 
    Given a quantum channel $\n_{A'\to A}$ and $\lim\alpha\in\{1,\infty\}$, we have 
    \begin{equation}
      {S}_\alpha^\varepsilon[\n] \le  \inf_{\op{\psi}\in\St(RA')}{S}_\alpha^{\varepsilon}(A|R)_{\n(\psi)}.
    \end{equation}
\end{nonumberedlemma}
\begin{proof}
We note that if ${\m}\in \mathcal{B}^\varepsilon[\n]$ for a quantum channel $\n_{A'\to A}$, then $ {\m}(\varphi_{RA'})\in \mathcal{B}^\varepsilon(\n(\varphi_{RA'}))$. This follows from the definition of the purified distance. For an arbitrary state $\varphi_{RA'}$ and $\alpha\in(1,\infty)$, we have
\begin{align}
{S}_\alpha^{\varepsilon}(A|R)_{\n(\varphi)}&=\sup_{\rho\in\mathcal{B}^\varepsilon(\n(\varphi_{RA'}))} {S}_\alpha(A|R)_{\rho}\\
&\ge \sup_{{\mathcal{M}}\in\mathcal{B}^\varepsilon[\n]} {S}_\alpha(A|R)_{{\mathcal{M}}(\varphi)}.
\end{align}
We take the infimum over pure states $\varphi\in\St(RA')$ on both sides, for $\lim\alpha\in\{1,\infty\}$,
    \begin{align}
\inf_{\op{\varphi}_{RA'}}{S}_\alpha^\varepsilon(A|R)_{\n(\varphi)}&\ge \inf_{\op{\varphi}_{RA'}}\sup_{{\m}\in\mathcal{B}^\varepsilon[\n]} {S}_\alpha(A|R)_{{\m}(\varphi)}\\
& \ge \sup_{{\m}\in\mathcal{B}^\varepsilon[\n]} \inf_{\op{\varphi}_{RA'}} {S}_\alpha(A|R)_{{\m}(\varphi)}\\
&= \sup_{{\m}\in\mathcal{B}^\varepsilon[\n]} {S}_\alpha[{\m}]\\
&= {S}_\alpha^\varepsilon[{\n}],
\end{align}
where we applied max-min inequality~\cite{Fan53,Dan66} to arrive at the second inequality.
\end{proof}

   \subsection{Proof of Lemma~\ref{lem:continuity}}\label{app:continuity_proof}
\begin{nonumberedlemma}
    Given two quantum channels $\n_{A'\to A}$ and $\m_{A'\to A}$ such that $\frac{1}{2}\Vert\n-\m\Vert_\diamond\le \delta$, the respective dynamical min-entropies satisfy  
    \begin{equation}
        \abs{S_{\min}[\n]-S_{\min}[\m]}\le \dfrac{1}{\ln 2}|A|\min\{|A|,|A'|\}\delta.
    \end{equation}
\end{nonumberedlemma}
\begin{proof}
    Note that $\dfrac{1}{2}\Vert\n-\m\Vert_\diamond\le \delta$ implies $\dfrac{1}{2}\Vert\n(\rho_{RA'})-\m(\rho_{RA'})\Vert_1 \le \delta$ for all states $\rho_{RA'}$, and it suffices to consider $R\simeq A'$ (in general $|R|\geq |A'|$). From the uniform continuity of conditional min-entropy as shown in~\cite[Lemma 21]{TCR10} and~\cite{MD22}, we have the following inequality for all $\rho\in\St(RA')$ and $|R|\geq |A'|$,
    \begin{equation}
        \abs{{S}_\infty(A|R)_{\n(\rho)}-{S}_\infty(A|R)_{\m(\rho)}} \le \dfrac{1}{\ln 2}|A|\min\{|A|,|A'|\}\delta.
    \end{equation}
Therefore, the difference between the channel min-entropies can be bounded as  
    \begin{align}
            {S}_\infty[\n]-{S}_{\infty}[\m] 
            &=\inf_{\op{\psi}\in\St(RA')} {S}_\infty(A|R)_{\n(\psi)}-\inf_{\op{\varphi}\in\St(RA')} {S}_\infty(A|R)_{\m(\varphi)}\\
            &\le {S}_\infty(A|R)_{\n(\varphi_0)}-{S}_\infty(A|R)_{\m(\varphi_0)}\\
            &\le \dfrac{1}{\ln 2}|A|\min\{|A|,|A'|\}\delta,
        \end{align}
where $\varphi_0\in\St(RA')$ is the optimal pure state that minimizes ${S}_\infty^\uparrow(A|R)_{\m(\varphi)}$.

On the other hand, we have 
    \begin{align}
            {S}_\infty[\n]-{S}_{\infty}[\m]
           &=\inf_{\op{\psi}\in\St(RA')} {S}_\infty(A|R)_{\n(\psi)}-\inf_{\op{\varphi}\in\St(RA')} {S}_\infty(A|R)_{\m(\varphi)}\\
            &\ge {S}_\infty(A|R)_{\n(\psi_0)}-{S}_\infty(A|R)_{\m(\psi_0)}\\
            &\ge -\dfrac{1}{\ln 2}|A|\min\{|A|,|A'|\}\delta,
        \end{align}
where $\psi_0\in\St(RA')$ is the optimal pure state that minimizes ${S}_\infty^\uparrow(A|R)_{\n(\psi)}$. This proves the lemma. 
\end{proof}    

\subsection{Proof of Theorem~\ref{thm:de-ch}}\label{app:proof_de_ch}
\begin{nonnumberedth}[Decoupling theorem for processes]
    Let $\mathcal{N}_{A'\to A}$ be a quantum channel, $\mathcal{T}_{A\to B}$ a completely positive map such that $\tr(\Gamma^\mathcal{T}_{AB})\leq |A|$, and $\varepsilon\in(0,1)$. The distance of the channel $\mathcal{N}$ post-processed by $\mathcal{T}\circ\U_A$, when $U_A$ is chosen uniformly at random from the Haar measure over the full unitary group $\mathbb{U}$ on $A$, with the uniformly randomizing channel $\mathcal{R}^{\pi}$ post-processed by $\mathcal{T}$ is upper bounded as
    \begin{align}
  \forall~\op{\psi}\in\St(RA'),\quad      \int_{\mathbb{U}(A)}\norm{\mathcal{T}_{A\to B}\circ\mathcal{U}_A\circ\n_{A'\to A}(\psi_{RA'})-\mathcal{T}_{A\to B}\circ\mathcal{R}^{\pi}_{A'\to A}(\psi_{RA'})}_1 \d\! U
        \le 2^{-\frac{1}{2}\left(S^\varepsilon_{\min}[\n]+S_{\min}^{\varepsilon}(A|B)_{\Phi^\mathcal{T}}\right)}+12\varepsilon,
        \end{align}
where $\Phi^\mathcal{T}_{AB}:=\frac{1}{|A|}\Gamma^{\mathcal{T}}_{AB}$ is a scaled Choi operator of $\mathcal{T}_{A\to B}$, $\mathcal{R}^\pi_{A'\to A}(\rho_{A'})=\pi_A \forall \rho\in\St(A')$, $\mathcal{U}_A(\cdot):=U_A(\cdot)U^\dag_A$.
\end{nonnumberedth}
\begin{proof}
    We begin our proof by noticing that for any state $\varphi_{AR}$ and any completely positive map $\mathcal{T}_{A\to B}$, the product state $\Phi^\mathcal{T}_B\otimes\varphi_R$ can be written as
    \begin{equation}
\Phi^\mathcal{T}_B\otimes\varphi_R=\mathcal{T}_{A\to B}\circ{\mathcal{R}^\pi}_{A\to A}(\varphi_{AR}).
    \end{equation}
Furthermore, for any trace-preserving map $\n_{A'\to A}$, we have ${\mathcal{R}^\pi}_{A\to A}\circ \n_{A'\to A}={\mathcal{R}^\pi}_{A'\to A}$. We employ the (one-shot) decoupling theorem for states. Substituting $\varphi_{AR}=\n_{A'\to A}(\psi_{RA'})$ in Eq.~\eqref{eq:decoup}, where $\psi_{RA'}$ is a pure state, we get
\begin{align}
    \int_{\mathbb{U}(A)}\norm{\mathcal{T}_{A\to B}\circ\mathcal{U}_A\circ\mathcal{N}(\psi_{RA'})-\mathcal{T}_{A\to B}\circ\mathcal{R}^\pi_{A'\to A}(\psi_{RA'})}_1 dU\le
        2^{-\frac{1}{2} (S^{\varepsilon}_{\min}(A|R)_{\n(\psi)} + S^{\varepsilon}_{\min}(A|B)_{\Phi^\mathcal{T}})} + 12 \varepsilon,
\end{align}
where the integration is with respect to the Haar measure over the full unitary group $\mathbb{U}$ on $A$. The r.h.s. in the above equation is bounded as follows, for all pure states $\op{\psi_{RA'}}$ we have
\begin{align}
2^{-\frac{1}{2} (S^{\varepsilon}_{\min}(A|R)_{\n(\psi_{RA'})} + S^{\varepsilon}_{\min}(A|B)_{\Phi^\mathcal{T}})} + 12 \varepsilon &\le     \sup_{\op{\varphi}_{RA'}}2^{-\frac{1}{2} (S^{\varepsilon}_{\min}(A|R)_{\n(\varphi_{RA'})} + S^{\varepsilon}_{\min}(A|B)_{\Phi^\mathcal{T}})} + 12 \varepsilon\\
        &=
        2^{-\frac{1}{2} (\inf_{\op{\varphi}_{RA'}}S^{\varepsilon}_{\min}(A|R)_{\n(\varphi)} + S^{\varepsilon}_{\min}(A|B)_{\Phi^\mathcal{T}})} + 12 \varepsilon\\
        &\le
       2^{-\frac{1}{2}(S^\varepsilon_{\min}[\n]+S_{\min}^{\varepsilon}(A|B)_{\Phi^\mathcal{T}})}+12\varepsilon.\label{eq:apprhs}
\end{align}
Here, we used the inequality $\inf_{\op{\psi}\in\St(RA')}S^{\varepsilon}_{\min}(A|R)_{\n(\psi)} \ge S^\varepsilon_{\min}[\n]$ from Eq.~\eqref{lem:smooth-entropy}. Therefore we have the result that for all pure states $\op{\psi}_{RA'}$, we have 
\begin{align}
\int_{\mathbb{U}(A)}\norm{\mathcal{T}_{A\to B}\circ\mathcal{U}_A\circ\n_{A'\to A}(\psi_{RA'})-\mathcal{T}_{A\to B}\circ\mathcal{R}^{\pi}_{A'\to A}(\psi_{RA'})}_1 \d\! U
        \le 2^{-\frac{1}{2}\left(S^\varepsilon_{\min}[\n]+S_{\min}^{\varepsilon}(A|B)_{\Phi^\mathcal{T}}\right)}+12\varepsilon.
\end{align}
\end{proof}

\subsection{Proof of Proposition~\ref{theo:channel_entropy}}\label{app:proof_channel_entropy}
\begin{nonnumberedprop}
    The min-entropy $S_{\min}[\n]$ of a quantum channel $\n_{A'\to A}$ is
\begin{align}
     S_{\min}[\n]
    &=-\sup_{\op{\psi}_{RA'}}\sup_{\m\in\Ch(R,\bar{A})} \log(|A|F(\m\otimes\n(\psi_{RA'}),\Phi_{A\bar{A}}))\label{eq:min-sing-a}\\
& =-\sup_{\rho\in\St(A')}\sup_{\sigma\in\St(E)}\log (|A|F(\mathcal{V}^\n_{A'\to AE}(\rho_{A'}),\pi_{A}\otimes\sigma_E)),\label{eq:min-dec-a}
\end{align}
where $\mathcal{V}^\n_{A'\to AE}$ is an isometric extension channel  of $\n_{A'\to A}$.
\end{nonnumberedprop}
\begin{proof}
    The first equality follows from~\cite{KRS09}, where it is shown that the min-conditional entropy can be written as a singlet fidelity,
\begin{equation}
    {S}_{\min}(A|R)_{\rho}=-\sup_{\m \in \mathrm{Ch(R,\bar{A})}} \log|A|F(\m_{R\to \bar{A}}(\rho_{AR}),\Phi_{A\bar{A}}).
\end{equation}
Then, the channel min-entropy is
\begin{align}
     {S}_{\min}[\n]
    &= \inf_{\op{\psi}\in\St(RA')}{S}_{\min}(A|R)_{\n(\psi)}\\
    &=-\sup_{\op{\psi}_{RA'}}\sup_{\m \in \mathrm{Ch(R,\bar{A})}} \log|A|F(\m_{R\to \bar{A}}\otimes\n_{A'\to A}(\psi_{RA'}),\Phi_{A\bar{A}}).
\end{align}
The second expression for the dynamical min-entropy is obtained utilizing the duality relation of the conditional min-entropy of states,
        \begin{align}
            {S}_{\min}[\n]&=\inf_{\op{\psi}\in\St(RA')} {S}_{\min}(A|R)_{\n(\psi)}\\
            &= \inf_{\op{\psi}\in\St(RA')}\left[ -{S}_{\frac{1}{2}}^\uparrow(A|E)_{\mathcal{V}^\n(\psi)}\right]\\
            &= \inf_{\op{\psi}_{RA'}}\left[- \sup_{\sigma\in\St(E)}\log F(\tr_R\circ\mathcal{V}^\n(\psi_{RA'}),\mathbbm{1}_A\otimes\sigma_E)\right]\\
            &=-\sup_{\op{\psi}_{RA'}}\sup_{\sigma\in\St(E)}\log F(\tr_R\circ\mathcal{V}^\n(\psi_{RA'}),\mathbbm{1}_A\otimes\sigma_E)\\
            &=-\sup_{\rho\in\St(A')}\sup_{\sigma\in\St(E)}\log |A|\log F(\mathcal{V}^\n(\rho_{A'}),\pi_{A}\otimes\sigma_E).
        \end{align}
\end{proof}
\subsection{Proof of Theorem~\ref{thm:workcost}}\label{app:proof_workcost}
\begin{nonnumberedth}
The adversarial erasure cost $W_{\mathrm{eras}}[A|E]_{\mathcal{N}}$ of a quantum channel $\n_{A'\to A}$ is bounded as
\begin{equation}
W_{\mathrm{eras}}[A|E]_{\mathcal{N}} \leq \left(-S^\varepsilon_{\min}[\n]+\Delta\right)k_BT\ln 2, 
\end{equation}
with the probability greater than $1-\delta$, where $\delta:=\sqrt{2^{-\Delta/2}+12\varepsilon}$, for all $\delta,\varepsilon> 0$.
\end{nonnumberedth}
\begin{proof}
Employing the results of~\cite{RAR+11}, we obtained Eq.~\eqref{eq:rar-bound}. Let us substitute $\delta=\sqrt{2\delta'}$ and $\Delta=-2\log(\delta^2-12\varepsilon)$ in Eq.~\eqref{eq:rar-bound}, we get with the probability greater than $1-\delta$ that
\begin{equation}
  W(A|E)_{\mathcal{V}^\n(\psi)}\le (-S^{\varepsilon}_{\min}(A|R)_{\n(\psi)}+\Delta)k_BT\ln 2,\label{eq:appiso}
\end{equation}
where $\mathcal{V}^\n$ is the isometric dilation of the channel $\n_{A'\to A}$ and $\psi_{RA'}$ is a pure state. Taking the supremum over all pure states $\psi\in\St(RA')$ on both sides, we get
    \begin{align} 
    \sup_{\op{\psi}_{RA'}}  W(A|E)_{\mathcal{V}^\n(\psi)}
      &\leq \sup_{\op{\psi}_{RA'}}\left(-S^{\varepsilon}_{\min}(A|R)_{\mathcal{N}(\psi)}+\Delta\right)\beta^{-1}\ln 2 \nonumber \\
      &= \left(-\inf_{\op{\psi}_{RA'}}S^{\varepsilon}_{\min}(A|R)_{\mathcal{N}(\psi)}+\Delta\right)\beta^{-1}\ln 2 \nonumber \\
      &\leq\left(-S^\varepsilon_{\min}[\n]+\Delta\right)\beta^{-1}\ln 2, 
    \end{align}
where $\beta^{-1}=k_BT$, and the final inequality follows from Lemma~\ref{lem:smooth-entropy}.
That is, 
\begin{equation}
    W_{\mathrm{eras}}[A|E]_{\mathcal{N}} \le (-S^{\varepsilon}_{\min}[\n]+\Delta)k_BT\ln 2, \label{eq:appisow}
\end{equation}
holds with probability $1-\delta$, as the bound~\eqref{eq:appisow} holds whenever the bound~\eqref{eq:appiso} holds.
\end{proof}

\textit{Interpretation of failure probability:} A crucial step in work extraction from the conditional erasure task for states is the ``compression of correlations" between the memory $E$ and the system $A$ into a $l$-qubit state that is $\delta$-close to a pure state in $AE$ using local unitary operations. The size of this pure state depends on the conditional max-entropy between the system $A$ and the memory $E$, as well as on $\delta$. Explicitly, given the $n$-qubit system $A$, the $l$-qubit state is $\delta$-close to a pure state where
\begin{align}\label{eq:2}
    l\ge n-S_{\max}^{\epsilon}(A|E)_\rho+2\log(\delta^2-12\epsilon).
\end{align}
This pure state can be used to extract $lk_bT\ln2$ amount of work with a failure probability of $\delta$.

Given a channel $\n_{A'\to A}$ with the isometric extension $\mathcal{V}^\n_{A'\to AE}$ and $R$ being the reference system, the adversarial erasure cost of $A$ conditioned on $E$ is given by Eq.~\eqref{eq:appiso}. The failure probability of this erasure task is less than $\delta$, where 
\begin{align}
    \delta^2&\ge 2^{(l-n+S_{\min}^\epsilon(A|E)_{\n(\psi)})/2}+12\epsilon.
\end{align}
Optimizing over all states $\psi_{RA'}$, we can say that there exists an adversarial erasure process of the channel whose cost is bounded by Eq.~\eqref{eq:appisow} except with a probability less than $\delta$ such that
\begin{align}
    \delta^2\ge 2^{(l-n+S^\epsilon_{\min}[\n])/2}+12\epsilon.
\end{align}

\subsection{Proof of Lemma~\ref{prop:eras_prep} }\label{app:proof_eras_prep}
\begin{nonumberedlemma}
    Given a state $\rho_{AB}$, the sum of the work cost of erasing and preparing the system $A$ conditioned on the system $B$, for error $\mu\in[0,1]$, is bounded from below as
\begin{equation}
        \widetilde{W}^{\mu}_{\mathrm{prep}}(A|B)_\rho+ \widetilde{W}^{\mu}_{\mathrm{eras}}(A|B)_\rho\ge \left[\log\left(1-\frac{\mu}{1-\mu^2}\right)-2\right]k_BT\ln 2,
    \end{equation}
and for $\mu=0$, i.e., zero-error erasure and preparation costs, we have
\begin{equation}
      \widetilde{W}^{0}_{\mathrm{prep}}(A|B)_\rho+ \widetilde{W}^{0}_{\mathrm{eras}}(A|B)_\rho\ge 0.
\end{equation}
\end{nonumberedlemma}
\begin{proof}
    The relation between the smooth relative entropy and the hypothesis testing for a state $\rho$ and a positive semidefinite operator $\sigma$ is given as~\cite{ABJT19,RLD25}
\begin{equation}
    D_{{H}}^{1-\mu^2-\delta}(\rho\Vert\sigma)-\log\frac{4}{\delta^2}\le {D}_\infty^{\mu}(\rho\Vert\sigma)-\log\frac{1}{1-\mu^2}~,
\end{equation}
where, $\mu\in(0,1)$ and $\delta\in (0,1-\mu^2)$. Taking $\delta=1-\mu-\mu^2$, we have
\begin{align}
    {D}^\mu_\infty (\rho\Vert\sigma)-D_{{H}}^\mu (\rho\Vert\sigma)\ge \log \left[\frac{1}{4}\left(1-\frac{\mu}{1-\mu^2}\right)\right].
\end{align}
Considering that $\rho\in\St(AB)$ and $\sigma_{AB}=\mathbbm{1}_A\otimes\rho_B$, we get
\begin{align}
 \log\left[\frac{1}{4}\left(1-\frac{\mu}{1-\mu^2}\right)\right]  
 &\leq  {D}^\mu_\infty (\rho_{AB}\Vert\mathbbm{1}_A\otimes\rho_B)- D_{{H}}^\mu (\rho_{AB}\Vert\mathbbm{1}_A\otimes\rho_B)\nonumber\\ 
    &\leq {D}^\mu_\infty (\rho\Vert\mathbbm{1}_A\otimes\rho_B)-\inf_{{\varphi}\in\St(B)}D_{{H}}^\mu(\rho\Vert\mathbbm{1}_A\otimes\varphi_B )\nonumber\\
    &  = -{S}^{\downarrow,\mu}_\infty(A|B)_{\rho}+S^{\uparrow,\mu}_{{H}}(A|B)_{\rho} \nonumber\\
    & = \frac{\beta}{\ln 2} \left(\widetilde{W}^{\mu}_{\mathrm{prep}}(A|B)_\rho+ \widetilde{W}^{\mu}_{\mathrm{eras}}(A|B)_\rho \right),
\end{align}
where $\beta=(k_BT)^{-1}$ is inverse temperature.

For $\mu=0$, we get
\begin{equation}
     \widetilde{W}^{0}_{\mathrm{prep}}(A|B)_\rho+ \widetilde{W}^{0}_{\mathrm{eras}}(A|B)_\rho\ge 0,
\end{equation}
using the inequality
\begin{equation}
    D_{H}^0(\rho\Vert\sigma)= \xoverline{D}_0(\rho\Vert\sigma)\le D_\infty(\rho\Vert\sigma).
\end{equation}
\end{proof}

\subsection{Proof of Proposition~\ref{prop:eras_prep_bounds}}\label{app:eras_prep_bounds_proof}
\begin{nonnumberedprop}
  Given a quantum channel $\n_{A'\to A}$ with an isometric extension $\mathcal{V}^\n_{A'\to AE}$ and a reference $R$, the work costs of preparing and erasing a channel when a reservoir (bath) is at a fixed temperature $T$, are bounded from the above as, for $\mu\in(0,1)$
\begin{align}
    \widetilde{W}^\mu_{\mathrm{prep}}[A|R]_\n &\le -S^\mu_{\min}[\n]k_BT\ln 2,\\
\widetilde{W}^\mu_{\mathrm{eras}}[A|E]_\n&\le  -\left(S_{\min}[\n]-\log(1-\mu)\right)k_BT\ln 2.
\end{align}
The bounds saturate for the zero-error work costs, i.e., when $\mu=0$, 
\begin{equation}
      \widetilde{W}^0_{\mathrm{prep}}[A|R]_\n=\widetilde{W}^0_{\mathrm{eras}}[A|E]_\n= -S_{\min}[\n]k_BT\ln 2.
\end{equation}
\end{nonnumberedprop}
\begin{proof}
We begin by proving the first bound. Given a pure state $\psi_{RA'}$ and a channel $\n_{A'\to A}$, from Eq.~\eqref{eq:prep_erase} (\cite{JGW25}) we have, for $\mu\in[0,1)$
\begin{equation}
        \sup_{\op{\psi}_{RA'}}  \widetilde{W}^{\mu}_{\mathrm{prep}}(A|R)_{\n(\psi)}= -\inf_{\op{\psi}_{RA'}}{S}_{\infty}^{\downarrow,\mu}(A|R)_{\n(\psi)}k_BT\ln 2.
\end{equation}
Using Lemma~\ref{lem:smooth-entropy}, we get the desired bound $\widetilde{W}^{\mu}_{\mathrm{prep}}[A|R]_\n \leq -S_{\infty}^{\mu}[\n]k_BT\ln 2$. 

To prove the second bound, we make use of the following inequality from~\cite{RLD25}, for $\mu\in[0,1)$ 
\begin{equation}
    D_{H}^{\mu}(\rho\Vert\sigma)\ge\xoverline{D}_{0}(\rho\Vert\sigma)+\log\dfrac{1}{1-\mu}.
\end{equation}
The detailed steps are as follows. For $\mu\in[0,1)$, we have
\begin{align}
   \frac{1}{k_BT\ln 2}\left(\sup_{\rho\in\St(A')}\widetilde{W}_{\mathrm{eras}}^\mu(A|E)_{\mathcal{V}^\n(\rho)}\right)
& =   \sup_{\rho_{A'}}S_{H}^{\uparrow,\mu}(A|E)_{\mathcal{V}^\n(\rho)} \\
        &=   \sup_{\rho\in\St(A')}\left[-\inf_{\sigma\in\St(E)} D_{H}^{\mu}(\mathcal{V}^\n(\rho_{A'})\Vert \mathbbm{1}_A\otimes\sigma_E)\right],\label{eq:prop2a}
\end{align}
then
\begin{align}
\sup_{\rho_{A'}}S_{H}^{\uparrow,\mu}(A|E)_{\mathcal{V}^\n(\rho)}
      &\leq   \sup_{\rho\in\St(A')}\left[-\inf_{\sigma\in\St(E)} \xoverline{D}_0(\mathcal{V}^\n(\rho_{A'})\Vert \mathbbm{1}_A\otimes\sigma_E)+\log(1-\mu)\right]\\
        & =   \sup_{\rho\in\St(A')}\xoverline{S}_0^{\uparrow}(A|E)_{\mathcal{V}^\n(\rho)}+\log(1-\mu)\\
        & = - \inf_{\op{\psi}\in\St(RA')}{S}_\infty^{\uparrow}(A|R)_{\n(\psi)}+\log(1-\mu)\\
        & = -S_\infty[\n]+\log(1-\mu).
    \end{align}  
   Finally, the saturation of the bounds at $\mu=0$ is evident, since
   \begin{align}
     S_{\min}[\n] & =\inf_{\op{\psi}\in\St(RA')} {S}_{\min}^{\downarrow,0}(A|R)_{\n(\psi)}\\ &=\inf_{\op{\psi}\in\St(RA')}{S}_{\min}^{\downarrow}(A|R)_{\n(\psi)}\\
    &= \inf_{\op{\psi}\in\St(RA')}\left(-\xoverline{S}_0(A|E)_{\mathcal{V}^\n(\psi)}\right)\\
    &= -\sup_{\rho\in\St(A')}\xoverline{S}_0(A|E)_{\mathcal{V}^\n(\rho)},
   \end{align}
   where $\mathcal{V}^{\n}_{A'\to AE}$ is an isometric extension channel  of $\n_{A'\to A}$.
\end{proof}
\section{Work extraction protocol}\label{app:protocol}
In this appendix, we will review the protocol of extracting work from a pure state~\cite{RAR+11}, where it is converted into a maximally mixed state. Let us have a finite $d$-dimensional system in a pure state $\ket{\varphi_0}$ with a fully degenerate Hamiltonian $H$. Let the set of orthogonal pure states $\{\ket{\varphi_i}\}_{i=0}^{N-1}$ form the basis of the $2^d$ Hilbert space, where $N=2^d-1$. Since the Hamiltonian is fully degenerate, all these states are also the eigenstates of $H$. The protocol is given as follows:
\begin{enumerate}
    \item Increase the energy levels of the states $\{\ket{\varphi_i}\}_{1}^{N-1}$ to a very high value $E\to \infty$, where the ground state energy is taken to be $E_0=0$. We can do this by manipulating the parameters of the Hamiltonian, e.g. the strength of the magnetic field. Since these energy levels are empty, this step does not lead to a work expenditure.
    \item Attach the system to a thermal bath at temperature $T$. In this configuration, the higher energy levels are occupied with probability $p(E)=N/(N+e^{\frac{E}{k_BT}})$, which vanishes in the limit $E\to \infty$ for any finite temperature.
    \item      Lowering the energy of the states $\{\ket{\varphi_i}\}_{1}^{N-1}$, their occupancy increases. Doing this isothermally yields the energy of $dk_BT\ln 2$ joules.
\end{enumerate}
The erasure process is exactly the opposite of the above protocol, where one converts the maximally mixed state into a pure state, resulting in the expenditure of $dk_BT\ln 2$ joules of work.
    \end{appendix}
\bibliography{thermo}   
\end{document}